\documentclass[a4paper,final]{article}
 
\usepackage{microtype}

\usepackage[obeyFinal,colorinlistoftodos]{todonotes}
\newcommand{\mytodo}[2]{\todo[size=\tiny, color=#1!50!white]{#2}\xspace}
\newcommand{\myinlinetodo}[2]{\todo[size=\small, color=#1!50!white, inline]{#2}\xspace}

\newcommand{\dkcom}[1]{\mytodo{purple}{#1}}
\newcommand{\fpcom}[1]{\mytodo{green}{#1}}

\newcommand{\osinline}[1]{\myinlinetodo{red}{#1}}

\title{Complexity of the Steiner Network Problem with Respect to the Number of Terminals}

\usepackage{tikz}
\usetikzlibrary{calc,positioning}
\usepackage{xspace}
\usepackage{enumerate}
\usepackage{comment}
\usepackage{amsmath,amsthm,amssymb}
\usepackage{hyperref}
\usepackage{fullpage}
\usepackage[utf8]{inputenc}
\usepackage{authblk}

\author[1,2]{Eduard Eiben\thanks{
		Supported by Pareto-Optimal Parameterized Algorithms (ERC Starting Grant 715744) and by the Austrian Science Fund (FWF, projects P26696 and W1255-N23).}}
\author[2,3]{Dušan Knop\thanks{Supported by grant NFR MULTIVAL and the project P202/12/G061 of GA ČR.}}
\author[2]{Fahad Panolan}
\author[4]{Ondřej~Suchý\thanks{Supported by grant 17-20065S of the Czech Science Foundation.}}
\affil[1]{Algorithms and Complexity Group, TU 
	Wien, Vienna, Austria
}
\affil[2]{Department of Informatics, University of Bergen, Bergen, Norway\newline
	\texttt{\{eduard.eiben, dusan.knop, fahad.panolan\}@uib.no}}
\affil[3]{Department of Applied Mathematics, Faculty of Mathematics and Physics,\newline Charles University, Prague, Czech Republic}
\affil[4]{Department of Theoretical Computer Science, Faculty of Information Technology,\newline Czech Technical University in Prague, Prague, Czech Republic\newline
	\texttt{ondrej.suchy@fit.cvut.cz}}

\newcommand{\OhOp}[1]{O\mathopen{}\mathclose\bgroup\left( #1 \aftergroup\egroup\right)}

\DeclareMathOperator{\tw}{\operatorname{tw}}

\DeclareMathOperator{\inc}{\operatorname{inc}}
\DeclareMathOperator{\adj}{\operatorname{adj}}
\newcommand{\QQQ}{\mathcal{Q}}

\newcommand{\degGr}[1]{\deg_{#1}}
\newcommand{\indegGr}[1]{\degGr{#1}^-}
\newcommand{\outdegGr}[1]{\degGr{#1}^+}
\DeclareMathOperator{\sym}{sym}

\usepackage{tabularx}
\newcommand{\prob}[3]{
\begin{center}
\begin{tabularx}{\textwidth}{|llXr|}
	\hline
	\hspace{1pt}\rule{0pt}{6pt}&\multicolumn{2}{l}{#1}&\\
	&{\bf Input:\enspace}&{#2}&\hspace*{1pt}\\
	&{\bf Question:\enspace}&{#3\rule{0pt}{1pt}}& \\
	\hline
\end{tabularx}
\end{center}
}

\theoremstyle{plain}
\newtheorem{theorem}{Theorem}
\newtheorem{lemma}[theorem]{Lemma}
\newtheorem{corollary}[theorem]{Corollary}

\newtheorem{claim}{Claim}
\newtheorem{observation}[theorem]{Observation}
\newtheorem{proposition}[theorem]{Proposition}
\newenvironment{claimproof}[1]{\par\noindent\textbf{Proof:}\space#1}{\hfill $\vartriangleleft$}
\theoremstyle{definition}
\newtheorem{definition}[theorem]{Definition}
\newtheorem{constr}{Construction}

\renewcommand{\subparagraph}[1]{\smallskip\noindent\textbf{#1}}

\usepackage{etoolbox}

\newcommand{\sv}[1]{}
\newcommand{\lv}[1]{#1}

\newcommand{\toappendix}[1]{%
	  #1
}

\newcommand{\appmark}{}

\begin{document}

\sloppy
\maketitle

\begin{abstract}
In the \textsc{Directed Steiner Network} problem we are 
given an arc-weighted digraph $G$, a set of terminals $T \subseteq V(G)$, and an (unweighted) directed request graph $R$ with $V(R)=T$. Our task is to output a subgraph $G' \subseteq G$ of the minimum cost such that there is a directed path from $s$ to $t$ in $G'$ for all $st \in A(R)$.

It is known that the problem can be solved in time $|V(G)|^{O(|A(R)|)}$ [Feldman\&Ruhl,  SIAM J. Comput. 2006] and cannot be solved in time $|V(G)|^{o(|A(R)|)}$ even if $G$ is planar, unless Exponential-Time Hypothesis (ETH) fails [Chitnis et al., SODA 2014]. However, as this reduction (and other reductions showing hardness of the problem) only shows that the problem cannot be solved in time $|V(G)|^{o(|T|)}$ unless ETH fails, there is a significant gap in the complexity with respect to $|T|$ in the exponent.

We show that \textsc{Directed Steiner Network} is solvable in time $f(R)\cdot |V(G)|^{O(c_g \cdot |T|)}$, where $c_g$ is a constant depending solely on the genus of $G$ and $f$ is a computable function.
We complement this result by showing that there is no $f(R)\cdot |V(G)|^{o(|T|^2/ \log |T|)}$ algorithm for any function $f$ for the problem on general graphs, unless ETH fails.
\end{abstract}


\section{Introduction}
\textsc{Steiner Tree} is one of the most fundamental and well studied problems in combinatorial optimization.
The input of \textsc{Steiner Tree} is an edge-weighted undirected graph $G$ and a set $T \subseteq V(G)$ of terminals. Here the task is to find a least cost connected subgraph $G'$ of $G$  containing all the terminals.
The problem is known to be NP-complete, and in fact was one of Karp’s original list~\cite{Karp72} of 21 NP-complete problems.  The problem is known to be APX-complete, even when the input graph is a complete graph and all edge weights are $1$ or $2$~\cite{BernP89}. On the other hand, the problem admits a constant factor approximation algorithm and current best approximation ratio is less than $1.39$~\cite{ByrkaGRS13}.
For an overview of the results and applications of \textsc{Steiner Tree}, the reader is referred to monographs~\cite{Cieslik98,HwangRW92,PromelS02}.


%

\textsc{Steiner Tree} is well studied in parameterized complexity.
The most natural parameter for the problem is the number of terminals.
The first FPT-algorithm for the problem is the $\OhOp{3^{|T|}\cdot n +2^{|T|} \cdot n^2 + n(n \log n+m)}$-time algorithm of Dreyfus and Wagner~\cite{DW72} (independently found by Levin~\cite{Levin71}) in 1970s; here and on $n$ denotes $|V(G)|$ and $m$ denotes $|E(G)|$.
This algorithm, as well as its later improvements~\cite{EricksonMV87,FKMRRW07,BHKK07} subsequently approaching the $O(2^{|T|}\operatorname{poly}(n+m))$ running time,
use exponential space.
The running time of $O(2^{|T|}\operatorname{poly}(n+m))$ is believed to be optimal assuming Set Cover Conjecture~\cite{CyganDLMNOPSW16}.
There has been many studies for designing algorithms with less space complexity.
A polynomial space single exponential time parameterized algorithm was designed recently by Fomin et al.~\cite{FominKLPS15}.
 

{\sc Steiner Tree} can be generalized to digraphs.  
There are many variants of Steiner-type problems on digraphs. The two most natural variants are \textsc{Directed Steiner Tree} (DST) and \textsc{Strongly Connected Steiner Subgraph} (SCSS). 
In DST, we are given an arc-weighted directed graph $G$, a set $T \subseteq V(G)$ of terminals, and root vertex $r\in V(G)$. Our task is to find a least cost subgraph $G'$ of $G$ such that for every $t\in T$, $t$ is reachable from $r$ in $G'$.
In SCSS, the input is an arc-weighted directed graph $G$ and a set $T \subseteq V(G)$ of terminals.
The task is to find a least cost subgraph $G'$ of $G$ such that  for every $s,t \in T$, there are directed paths from $s$ to $t$ and $t$ to $s$ in $G'$. That is, $G'$ is a least cost strongly connected subgraph containing all the terminals.
A common generalization of DST and SCSS is \textsc{Directed Steiner Network} (DSN).
In DSN, we are given an arc-weighted digraph $G$, a set $T \subseteq V(G)$ of terminals and a digraph $R$ on $T$. The task is to find a least cost subgraph $H$ of $G$ which realizes all paths prescribed by the arcs of $R$.
That is, for any arc $st\in A(R)$, there is a directed path from $s$ to $t$ in $H$.
Observe that in DSN, request graphs $R$ and $R'$ yield the same set of solutions if their transitive closures are the same.
DST is a special case of DSN where $R$ is single out-tree on $T\cup \{r\}$ with $r$ being the root and $T$ being the set of leaves. Similarly, SCSS is a special case of DSN where $R$ is a single directed cycle on the set of terminals $T$.


Existence of an $\alpha$-approximation algorithm for DST implies a $2\alpha$-approximation algorithm for SCSS because of the following simple observation. The union of an in-tree and an out-tree from one fixed terminal in $T$ yields a strongly connected subgraph containing $T$.   
The best known approximation ratio in polynomial time for DST and SCSS is $O(|T|^\varepsilon)$ for any $\varepsilon >0$~\cite{CCCD+99}.
The same paper also yields $O(\log^2 |T|)$-approximation algorithm in quasi-polynomial time. A result of Halperin and Krauthgamer~\cite{HalperinK03} implies that DST and SCSS have no $O(log^{2-\varepsilon} n)$-approximation for any $\varepsilon > 0$, unless NP has quasi-polynomial time Las Vegas algorithms.
The best known approximation algorithm for DSN is by Chekuri et al.~\cite{CEGS08} with an approximation factor of ~$O(|T|^{1/2+\epsilon})$ for any $\epsilon>0$.
On the other hand, DSN cannot be approximated to within a factor of~$O(2^{\log^{1-\epsilon}n})$ for any fixed~$\epsilon>0$, unless~NP $\subseteq$ TIME$(2^{\text{polylog}(n)})$~\cite{DK99}.



Using essentially the same techniques as for \textsc{Steiner Tree}~\cite{DW72}, one can show that 
there is  an $\OhOp{3^{|T|}\cdot n +2^{|T|} \cdot n^2 + n(n \log n+m)}$ time algorithm for DST. On the other hand, Guo et al.~~\cite{GuoNS11} showed that SCSS parameterized by $\vert T \vert$ is W[1]-hard.
That is, there is no $f(|T|)\cdot n^{O(1)}$ time algorithm for SCSS for any function $f$, unless FPT=W[1]. 
Later a stronger lower bound has been shown by Chitnis et al.~\cite{ChitnisHM14}. They showed that 
in fact there is no $f(|T|)n^{o(|T|/\log |T|)}$ algorithm for SCSS for any function $f$, unless Exponential Time Hypothesis (ETH) of Impagliazzo and Paturi~\cite{ImpagliazzoP01} fails.
This stimulated the research on DSN for restricted classes of request graphs~\cite{Suchy16, FeldmannM16}.

\osinline{we take perpendicular approach by leaving the request graph unrestricted and restricting the host graph}

As DSN is a generalization of SCSS and SCSS is W[1]-hard, DSN is also W[1]-hard.
On the positive side, Feldman and Ruhl~\cite{FR06} showed that DSN can be solved in $n^{O(|A(R)|)}$ time.
An independent algorithm with similar running time also follows from the classification work of Feldmann and Marx~\cite{FeldmannM16}.
Complementing these results Chitnis et al.~\cite{ChitnisHM14} showed that DSN cannot be solved in $f(R)n^{o(|A(R)|)}$ time for any function $f$, even when restricting the host graph $G$ to be planar and all arc weights equal to one, unless ETH fails.
However, in this reduction (as well as in the reduction given for SCSS), the number of arcs of the request graph $|A(R)|$ is linear in the number of terminals $|T|=|V(R|)$.
Hence, viewed in terms of the number of terminals, this lower bound only implies that there is no $f(|T|)n^{o(|T|)}$ time algorithm for any function $f$, unless ETH fails.
But both the known algorithms have runtime $n^{\Theta(|T|^2)}$ in worst case, leaving a significant gap between the upper and the lower bound for~DSN.
In this work we try to fill this gap. We prove the following upper and lower bounds.

\begin{theorem}\label{thm:DSN_surfaces_for_terminals}
There is an algorithm for DSN which correctly solves any instance $(G,R)$ in $f(R)\cdot n^{O(c_g \cdot |T|)}$ time, where $c_g=20^{16g+12}g$, $f$ is some computable function, and $g$ is the genus of the graph $G$.
\end{theorem}

\begin{theorem}\label{thm:gen_hard}
There is no $f(R)\cdot n^{o(|T|^2/ \log |T|)}$ time algorithm for DSN on general graphs for any function $f$, unless ETH fails.
\end{theorem}

\subparagraph{Our methods.}
Let $(G,R)$ be an instance of DSN and let $H$ be a least cost subgraph of $G$ which realizes all paths prescribed by the arcs of $R$ (call it \emph{an optimum solution}). 
From the classification work of DSN, by Feldmann and Marx~\cite{FeldmannM16}, it follows that if the treewidth of $H$ is $\omega$, then there is an algorithm $\mathbb{A}$ for solving DSN, running in time $f(R,\omega)n^{O(\omega)}$, where $f$ is a computable function.
Towards proving the existence of $f(R)\cdot |V(G)|^{O(c_g \cdot |T|)}$ time algorithm  where $G$ is embeddable into a surface of genus $g$, we show that if an instance $(G,R)$ of DSN has at least one solution, then there is a solution $H$ with treewidth bounded by ${O(c_g \cdot |T|)}$.
Then our result follows. As $H$ is subgraph of $G$, $H$ is also embeddable into a surface of genus $g$.
To prove that $H$ has treewidth ${O(c_g \cdot |T|)}$, we derive a graph $H'$ such that
\begin{itemize}
  \item $H'$ is also embeddable into a surface of genus $g$,
  \item $\tw(H) = \OhOp{\tw(H')}$, and
  \item the diameter of $H'$ is $O(\vert T \vert)$.
\end{itemize}
Eppstein~\cite{Eppstein00} proved that if a graph $H'$ have genus $g$ and diameter $D$, then $H'$ has treewidth $\OhOp{gD}$. Thus we conclude that $H$ has treewidth ${O(c_g \cdot |T|)}$ and our result follows using the algorithm $\mathbb{A}$ of Feldmann and Marx~\cite{FeldmannM16}.

\subparagraph{Negative Results.}
Towards our lower bound result, we give a reduction from \textsc{Partitioned Subgraph Isomorphism} (PSI).
In PSI, we are given two undirected graphs $G$ and $H$ along with a function $\psi \colon V(G) \rightarrow V(H)$.
Our task is to test whether $H$ is isomorphic to a subgraph $G'$ of $G$ respecting the map $\psi$.
That is, we want to test the existence of an injective map $\phi \colon V(H)\rightarrow V(G)$ such that $\{\phi(u),\phi(v)\}\in E(G)$ for all $\{u,v\}\in E(H)$ and $\phi(u)\in \psi^{-1}(u)$ for all $u\in V(H)$.

\section{Preliminaries}
\sv{\toappendix{
\subsection{Additions to Preliminaries}
}}
For a positive integer $n$, we use $[n]$ to denote the set $\{1,\ldots,n\}$.  
We consider simple directed graphs and use mostly standard notation that can be found for example in the textbook by Diestel~\cite{Diestel17}.
Let $G$ be a directed graph with vertex set $V(G)$ and arc set $A(G)$.
For vertices $u,v \in V(G)$ the arc from $u$ to $v$ is denoted by $uv \in A$ or $(u,v) \in A$.
A \emph{walk} $P = (p_0, \ldots, p_\ell)$ of length $\ell$ in $G$ is a tuple of vertices, that is $p_i \in V(G)$ for all $1 \le i \le \ell$, such that $p_ip_{i+1} \in A(G)$ for all $0 \le i < \ell$.
A \emph{directed path} $P = (p_0, \ldots, p_\ell)$ in $G$ is a walk of length $\ell$ with all vertices distinct, that is $p_i \neq p_j$ for all $0 \le i < j \le \ell$.
We let $V(P)= \{p_0, \ldots, p_\ell\}$.
We say that the path $P$ is from $p_0$ to $p_\ell$, call these vertices the \emph{endpoints} of $P$ while the other vertices are called \emph{internal} (we denote the set of these vertices by $\mathring{P}$).
Path $P$ is between $u$ and $v$ if it is either from $u$ to $v$ or from $v$ to $u$. 
Let $W$ be a set of vertices, we say that a path $Q$ is a \emph{$W$-avoiding path} if $\mathring{Q} \cap W = \emptyset$; if $P$ is a path we say that $Q$ \emph{$P$-avoiding path} if it is $V(P)$-avoiding path.
Let $P$ be a walk from $u$ to $v$ and let $Q$ be a walk from $v$ to $w$ by $P \circ Q$ we denote the \emph{concatenation} of $P$ and $Q$ and define it to be the walk from $u$ to $w$ that follows $P$ from $u$ to $v$ and then follows $Q$ from $v$ to $w$.
Let $P = (p_0, \ldots, p_\ell)$ be a directed path and $u, v \in V(P)$. We write $u\le_P v$ if $u$ is before $v$ on $P$ or in other words $u=p_i$ and $v=p_j$ such that $i\le j$.
For a path $P = (p_0, \ldots, p_\ell)$ and two vertices $u,v \in V(P)$ with $u\le_P v$ the \emph{subpath of $P$ between $u$ and $v$}, denoted $u[P]v$, is the path $(p_i, \ldots, p_j)$, where $p_i = u$ and $p_j = v$.
For a vertex $v \in V(G)$ its \emph{in-degree} is defined as $\indegGr{G}(v) = \left| \left\{ u \in V \mid uv \in A(G) \right\} \right|$.
The \emph{out-degree} of $v$ is $\outdegGr{G}(v) = \left| \left\{ u \in V \mid vu \in A(G) \right\} \right|$.
Finally, the \emph{total degree} of $v$ is $\degGr{G}(v) = \indegGr{G}(v) + \outdegGr{G}(v)$.
If the graph $G$ is clear from the context we drop the subscript~$G$.
We use $\sym(G)$ to denote the underlying undirected graph of a directed graph $G$.
To \emph{subdivide an arc $e\in A(G)$} is to delete $e = uv$, add a new vertex $w$, and add arcs $uw, wv$.
We say that $H$ is a \emph{subdivision} of $G$ if it can be obtained by repeated subdivision of arcs of $G$, that is there exist graphs $G = G_0, \ldots, G_n = H$ such that $G_{i + 1}$ is the result of arc subdivision of $G_i$.

We consider the following problem:
\prob{\textsc{Directed Steiner Network} (DSN)}
{An arc-weighted directed graph $G$ and an (unweighted) directed graph $R$ such that $V(R) \subseteq V(G)$.}
{Find a minimum-cost subgraph $H$ of $G$ in which there is a path from $s$ to $t$ for every $st \in A(R)$.}
The problem is also called \textsc{Directed Steiner Forest} or \textsc{Point-to-Point Connection}.
We only consider positive weights on arcs, as allowing negative-weight arcs would lead to a different problem.
We call a subgraph $H$ of $G$ \emph{a solution} to the instance $(G, R)$ of DSN if $H$ contains a path from $s$ to $t$ for every $st \in A(R)$. 
Moreover, we say that $H$ is an \emph{inclusion-minimal} solution to $R$, if $H$ is a solution for some instance $(G,R)$, but for every $e\in A(H)$, $H-e$ is not.
%
Note that an optimum solution (one with the least sum of weights) is necessarily inclusion-minimal, as we assume positive weights.
Through the paper we denote $p= |A(R)|$, $T=V(R)$, $q=|T|$, $n=|V(G)|$, and $H$ a solution to $(G, R)$.

Let $\mathcal{R}$ be a class of graphs. The $\mathcal{R}$-DSN problem is such a variant of DSN where it is promised that the graph $R$ belongs to the class $\mathcal{R}$.

\begin{proposition}[{Feldmann and Marx~\cite[Theorem 5]{FeldmannM16} (see also~\cite{FeldmannM17})}]\label{prop:solution_tw_algorithm}
Let an instance of $\mathcal{R}$-DSN be given by a graph $G$ with $n$ vertices, and 
a pattern $R$ on $k$ terminals with vertex cover number $\tau$. If the 
optimum solution to $R$ in $G$ has treewidth~$\omega$, then the optimum can be 
computed in time $2^{O(k+\max\{\omega^2,\tau\omega\log\omega\})}n^{O(\omega)}$.
\end{proposition}


\begin{proposition}[Demaine, Hajiaghayi, and Kawarabayashi~\cite{DemaineHK09}]\label{thm:MinorFreeTW}
	Suppose $G$ is a graph with no $K_{3,k}$-minor. If the treewidth is at least
	$20^{4k}r$, then $G$ has an $r \times r$ grid minor.

\end{proposition}

\begin{proposition}[\cite{GrossTucker87}]\label{pro:genus-minor}
For a graph $G$, let $\gamma(G)$ and $\tilde\gamma(G)$ denote the orientable and nonorientable genus of $G$, respectively.
 If $G$ is a graph such that $\gamma(G)\le g$ and $\tilde\gamma(G)\le h$ for some constants $g$ and $h$, then $G$ does not contain $K_{3,4g+3}$ nor $K_{3,2h+3}$ as a minor. 
\end{proposition}

\begin{proposition}[{Eppstein~\cite[Theorem 2]{Eppstein00}}]\label{prop:genus_tw}
	Let $G$ be a graph of genus $g$ and diameter $D$.
	Then G has treewidth $\OhOp{gD}$.
\end{proposition}

For a more detailed treatment of topological graph theory the reader
is referred to \cite{GrossTucker87}.

\toappendix{
\subparagraph{$t$-Boundaried Graphs and Gluing.}
A \emph{$t$-boundaried graph} is a graph $G$ and a set $B \subseteq  V(G)$ of size at most $t$ with each vertex $v \in  B$ having a label $G(v) \in \{1, \ldots , t\}$. Each vertex in $B$ has a unique label. We refer to
$B$ as the boundary of $G$. For a $t$-boundaried graph $G$ the function $\delta(G)$ returns the boundary of $G$. Two $t$-boundaried graphs $G_1$ and $G_2$ can be \emph{glued} together to form a graph $G = G_1\oplus G_2$. The gluing operation takes the disjoint union of $G_1$ and $G_2$ and identifies the vertices of $\delta(G_1)$ and $\delta(G_2)$ with the same label.

A $t$-boundaried graph H is a minor of a $t$-boundaried graph
$G$ if (a $t$-boundaried graph isomorphic to) $H$ can be obtained from $G$ by deleting vertices or
edges or contracting edges, but never contracting edges with both endpoints being boundary
vertices.
For more details see e.g.~\cite{FominLMS12}.
}

\toappendix{
\subparagraph{Monadic Second Order Logic.}
The syntax of Monadic second order logic (MSO) includes the logical connectives $\vee, \wedge, \neg, \Rightarrow, \Leftrightarrow$, variables for
vertices, edges, sets of vertices and sets of edges, the quantifiers $\forall, \exists$ that can be applied to these variables, and the following five binary relations:
\begin{enumerate}
	\item $u\in U$ where $u$ is a vertex variable and $U$ is a vertex set variable;
	\item $d \in D$ where $d$ is an edge variable and $D$ is an edge set variable;
	\item $\inc(d, u)$, where $d$ is an edge variable, $u$ is a vertex variable; and the interpretation is that the edge $d$ is incident on the vertex $u$;
	\item $\adj(u, v)$, where $u$ and $v$ are vertex variables, and the interpretation is that $u$ and $v$ are adjacent;
	\item equality of variables representing vertices, edges, set of vertices and set of edges.
\end{enumerate}
Many common graph-theoretic notions such as vertex degree, connectivity, planarity, outerplanarity, 
being acyclic, and so on, can be expressed in MSO, as can be seen from introductory expositions~\cite{LibkinFMT}.
}

\section{Solving DSN on a Fixed Surface}\label{sec:DSN_on_surface}
\sv{\toappendix{
\section{Additions to Section~\ref{sec:DSN_on_surface}}
}}
We fix an instance $(G,R)$ of DSN with $g$ being the genus of $G$, $T = V(R)$, $q = |T|$, and $H$ is an inclusion-minimal solution to $(G,R)$.
In the rest of this section we assume $g$ is fixed.
Note that the genus of $H$ is at most $g$.

The goal of this section is to show the following theorem.
\begin{theorem}\label{thm:H_has_bounded_tw}
Treewidth of $H$ is $\OhOp{20^{4(4g+3)}g \cdot q}$.
\end{theorem}
With this theorem at hand Theorem~\ref{thm:DSN_surfaces_for_terminals} follows immediately from Proposition~\ref{prop:solution_tw_algorithm}.

\subparagraph{Reversing Arcs -- Symmetry.}
Let $\overleftarrow{G}$, $\overleftarrow{H}$, and $\overleftarrow{R}$ be the directed graphs we obtain from $G$, $H$, and $R$, respectively, by reversing all the arcs.
That is, for example, $\overleftarrow{G}$ contains an arc $uv$, if and only if $G$ contains the arc $vu$. Note that there is a one-to-one correspondence between an $s$-$t$ path in $H$ and a $t$-$s$ path in $\overleftarrow{H}$.
Hence, if $H$ is an optimum solution to the instance $(G,R)$, then $\overleftarrow{H}$ is an optimum solution to the instance $\bigl(\overleftarrow{G}, \overleftarrow{R} \bigr)$.
The importance of $\overleftarrow{G}, \overleftarrow{H}, \overleftarrow{R}$ is that every lemma holds in both $H, R$ and $\overleftarrow{H}, \overleftarrow{R}$.
In this way we obtain symmetric results without reproving everything twice.

\begin{observation}\label{obs:adjustingR}
Let $R'$ be a directed graph with vertex set $T$ and $st\in A(R')$ for $s,t \in T$ if and only if there is a $T$-avoiding path in $H$ from $s$ to $t$.
Then $H$ is an inclusion-minimal solution to $R'$. \fpcom{$(G,R')$ ??}\dkcom{no -- it is property of $H$ and $R$ no matter which $G$.}
\end{observation}
From now on we replace $R$ with $R'$.

Let $H_1, H_2$ be two directed graphs.
We say that the pair $\left( H_1, H_2 \right)$ is an \emph{admissible pair} if
   $H_2$ is embeddable into the same surface as $H_1$ and 
   $\tw(H_1) = \OhOp{\tw(H_2)}$. 

\lv{
\begin{observation}\label{obs:admissibleTransitivity}
Let $\left( H_1, H_2 \right)$ and $\left( H_2, H_3 \right)$ be admissible pairs.
Then $\left( H_1, H_3 \right)$ is an admissible pair.
\end{observation}
}

\subparagraph{Overview of the Proof of Theorem~\ref{thm:H_has_bounded_tw}.}
To prove Theorem~\ref{thm:H_has_bounded_tw} we transform the solution graph $H$ into a graph $H'$ containing all terminals and preserving all terminal-to-terminal connections such that $\left( H, H' \right)$ is an admissible pair and $H'$ has bounded diameter (and thus by Proposition~\ref{prop:genus_tw} has bounded treewidth).
We do this by exploiting that a terminal to terminal path in $H$ contains only $O(q)$, so called, important and marked vertices.
Furthermore, a subpart of the solution ``between'' two consecutive marked or important vertices has constant treewidth and contains few vertices with arcs to the rest of the solution $H$.
This allows us to replace this part of a solution to constant size while preserving genus and all terminal-to-terminal connections.
Thus obtaining the graph $H'$ of bounded diameter.

\begin{lemma}\label{lem:degreeTwoVertices}
There exists directed graph $H_{\ge 3}$ such that every non-terminal vertex in $H_{\ge 3}$ has at least three neighbors and $\left( H, H_{\ge 3}\right)$ is an admissible pair.
\end{lemma}
\begin{proof}[Proof Sketch]
We exhaustively repeat the following until we cannot apply it anymore.
Let $v$ be a non-terminal and let $u,w$ be the two neighbors of $v$.
Note that $v$ cannot have degree one, since $H$ is a minimal solution.
We delete $v$ from $H$ and add an arc $uw$ if both $uv$ and $vw$ were in $H$, similarly for an arc $wu$.
Denote the resulting graph $H_{\ge 3}$.
\end{proof}

From now on we replace $H$ with $H_{\ge 3}$.

\subsection{Important and Marked Vertices}
For a fixed $T$-avoiding directed path $P$ in $H$ between two terminals $s$ and $t$, we say that vertex $u\in V(P)$ is \emph{important with respect to $P$} if there is a $P$-avoiding directed path from some terminal not on $P$ to $u$ or from $u$ to some terminal not on $P$.
Let $I_P$ denote the set of vertices important with respect to $P$.
Let $I$ be the union of important vertices over all paths. 

Let $s,t\in T$ and $P = (s=p_1, \ldots, p_r =t)$ be fixed for the rest of this subsection.

\begin{lemma}\label{lem:important_vertices_on_P}
There are at most $2q-2$ important vertices on $P$.
Moreover, there exists a function $g_P \colon I_P \to T$ with $\left| g_P^{-1}(x) \right| \le 2$ for every $x\in T$ such that for every $v\in I_P$ there is either $v$-$g(v)$ or $g(v)$-$v$ directed $\left(V(P)\cup T\right)$-avoiding path.
\end{lemma}
\begin{proof}
We bound the number of important vertices by inspecting the interaction between the path $P$ and other paths in the solution $H$.
We claim that every important vertex on the path $P$ can be labelled by a pair consisting of a terminal $x \in T$ and a direction $\left\{ \leftarrow, \rightarrow \right\}$.
In order to prove this, we label an important vertex $v\in V(P)$ by $(x, \leftarrow)$ if there is a directed $P$-avoiding path from a terminal $x$ to $v$ in $H$ and $v$ is the closest to $s$ among all such vertices of $P$.
Similarly, we label an important vertex $v\in V(P)$ by $(x, \rightarrow)$ if there is a directed $P$-avoiding path from $v$ to a terminal $x$ in $H$ and $v$ is the closest to $t$ among all such vertices of $P$.
\begin{claim}\appmark\label{clm:important_receive_a_label}
Every important vertex received some label.
\end{claim}
It follows from the above claim that the number of important vertices is bounded by the possible number of labels with is $2q - 2$.
This is because $(s, \leftarrow)$ and $(t, \rightarrow)$ labels are never assigned to any important vertex of $P$ (as they would be assigned to $s$ and $t$, respectively).

\toappendix{
\begin{proof}[Proof of Claim~\ref{clm:important_receive_a_label}]

Suppose for contradiction that there is an important vertex $v \in V(P)$ with no assigned label.
Assume that there is a $P$-avoiding directed path from some terminal not on $P$ to $v$ (a symmetric argument works if there is a $P$-avoiding directed path from $v$ to some terminal not on $P$).
There must be an arc $uv$ in $H$ such that $u \notin V(P)$.

Suppose first that there is an arc $uv$ in $H$ such that $u \notin V(P)$ but there is a $P$-avoiding path $Q$ from $u$ to some vertex $v'$ on $P$ with $v' <_P v$. As any path using the arc $uv$ can be detoured to use the path $Q$ and then $v'[P]v$, removing the arc $uv$ from $H$ does not violate any connection requested in $R$, contradicting the inclusion-minimality of $H$. 

Suppose next that there is an arc $uv$ in $H$ such that $u \notin V(P)$ and there is no $P$-avoiding path from a vertex $v'$ on $P$ with $v' >_P v$ to $u$. We claim that the arc $uv$ can be removed from $H$ without violating any connection requested in $R$, contradicting the inclusion-minimality of $H$.
Suppose that $s'$ is a terminal such that there is a path $Q$ from $s'$ to $v$ in $H$ but not in $H - uv$.
The arc $uv$ is the last arc of $Q$.
If $\mathring{Q}$ contained a vertex of $P$, then let $v'$ be the last such vertex. 
If $v' <_P v$, then we would have the walk $s'[Q]v' \circ v'[P]v$ from $s'$ to $v$ in $H - uv$, which can be shortened to a path from $s'$ to $v$ in $H - uv$, contradicting our assumptions. If $v' >_P v$, then we would have $P$-avoiding path from a vertex $v'$ on $P$ with $v' >_P v$ to $u$, again contradicting our assumptions. Hence $Q$ is $P$-avoiding and witnesses that $v$ is important.
As $v$ did not receive the label $(s',\leftarrow)$ via the above described procedure,
there is an important vertex $v'$ of $P$ labeled $(s', \leftarrow)$ with $v' <_P v$.
But then there is a $P$-avoiding path $Q'$ from $s'$ to $v'$ in $H$ and, as it does not go trough $v$, it is preserved in $H - uv$. 
But then $Q' \circ v'[P]v$ is a path from $s'$ to $v$ in $H - uv$, contradicting our assumption.

Finally let $U$ be the set of vertices $u \notin V(P)$ having a $P$-avoiding path from $u$ to $v$, not having a $P$-avoiding path from $u$ to any $v'$ on $P$ with $v' <_P v$ and having a $P$-avoiding path from a vertex $v'$ on $P$ with $v' >_P v$ to $u$. By the previous two cases, all in-neighbors of $v$ not on $P$ are in $U$.  Let a terminal $y$ and a $P$-avoiding path $Q$ from $y$ to $v$ witness that $v$ is important. Since $v$ did not receive the label $(y,\leftarrow)$ via the above described procedure, $y$ is not in $U$. Nevertheless the second to last vertex of $Q$ is in $U$. 
Let $ww'$ be the arc of this path such that $w \notin U$ and $w' \in U$. 
We claim that the arc $ww'$ can be removed from $H$ without violating any connection requested in $R$, contradicting the inclusion-minimality of $H$.

Suppose that $s'$ is a terminal such that there is a path from $s'$ to $w'$ in $H$ but not in $H - ww'$.
As $v$ did not receive the label $(s',\leftarrow)$ via the above described procedure,
there is an important vertex $v'$ of $P$ labeled $(s', \leftarrow)$ with $v' <_P v$.
But then there is a $P$-avoiding path $Q'$ from $s'$ to $v'$ in $H$. 
This path cannot go through $w'$, since $w'$ is in $U$ and, thus, does not have a $P$-avoiding path to any $v'$ on $P$ with $v' <_P v$. 
Therefore, $Q'$ is preserved in $H - ww'$.
Moreover, since $w' \in U$, there is a $P$-avoiding path $Q''$ from some vertex $v''$ with $v''>_P v$ to $w'$. 
But then $Q' \circ v'[P]v'' \circ Q''$ is a path from $s'$ to $w'$ in $H - ww'$, contradicting our assumption.
\end{proof}
}

As to the moreover part, it follows from the labeling procedure that if an important vertex $v$ received a label $(\leftarrow, x)$, then there is a $P$-avoiding path $Q$ in $H$ from $x$ to $v$. If this path contains another terminal, then let $y$ be the terminal closest to $v$ on $Q$. We claim that $v$ also received the label $(\leftarrow, y)$. If not, then there would be another vertex $v'$ on $P$ with this label with $v'<_P v$ and a $P$-avoiding path $Q'$ from $y$ to $v'$. But then $x[Q]y \circ y[Q']v'$ is a walk that can be shortened to a $P$-avoiding path from $x$ to $v'$, contradicting $x$ receiving the label $(\leftarrow, x)$. Hence, each important vertex has a $(V(P) \cup T)$-avoiding path to or from some terminal, such that it has a label of that terminal. To prove the moreover part it remains to set $g_P(v)$ to any such terminal.
\end{proof}

\lv{
The following lemma shows that the influence of non-important vertices is limited.
}

\begin{lemma}\appmark\label{lem:nonimportant}
 If $v$ is a vertex in $V(P) \setminus I_P$, then its out-degree is at most 2. Moreover, if $u$ is its out-neighbor not on $P$, then there is a $P$-avoiding path from $u$ to some vertex $v' \in V(P)$ with $v' <_P v$.
\end{lemma}
\toappendix{
\begin{proof}[Proof of Lemma~\ref{lem:nonimportant}]
Suppose $v \in V(P) \setminus I_P$. If $v$ has out-degree 1, then we are done. Otherwise let $v'$ be the out-neighbor of $v$ on $P$ and let $X$ be the set of vertices different from $v$ and $v'$ having a $P$-avoiding path from $v$. Since $v$ is not important, there is no terminal in $X$. 
 
Each vertex $x \in X$ is on some $s'$-$t'$-path in $H$ for some $s't' \in A(R)$, as otherwise we could remove $x$ from $H$ to obtain smaller solution, contradicting the inclusion-minimality of $H$. It follows that each such path through $x$ must intersect $P$ (either $x \in V(P)$ or some vertex after $x$ on the path is in $V(P)$). 

Let $u'$ be the vertex in $X \cap V(P)$ closest to $s$. If $u' <_P v$ then let $u \in X$ be the out-neighbor of $v$ on any $P$-avoiding path $Q$ from $v$ to $u'$. Otherwise do not denote any vertex as $u$.
If there is an out-neighbor $w$ of $v$ different from $v'$ and $u$, then we claim that the arc $vw$ can be removed from $H$ without violating any connection requested in $R$, contradicting the inclusion-minimality of $H$.

Suppose that $t'$ is a terminal such that there is a path $Q'$ from $v$ to $t'$ in $H$ but not in $H - vw$.
That is, the path uses $vw$ as its first arc and must intersect $P$. Let $w'$ be the first vertex of $P$ on $Q$ after $v$. If $w'>_P v$, then $v[P]w' \circ w'[Q']t'$ is a path from $v$ to $t'$ in $H - vw$, contradicting our assumptions. If $w'<_P v$, then $w'>_P u'$ by the choice of $u'$. But then $v[Q]u' \circ u'[P]w' \circ w'[Q']t'$ is a path from $v$ to $t'$ in $H - vw$, again contradicting our assumptions. 
This finishes the proof of the lemma.
\end{proof}
}

The following expresses that in order to bound the diameter of $H'$ it is enough to bound the length of the path $P$ linearly in $|I_P|$.

\begin{lemma}\label{lem:small_distances}
If for every $\bar{s}\bar{t}\in A(R)$ there is a $T$-avoiding path $\tilde{P}$ in $H$ of length at most $c\cdot |I_{\tilde{P}}|$, for some constant $c$, then the distance between any two terminal vertices in the underlying undirected graph $\sym(H)$ of $H$ is at most $8cq$.
\end{lemma}
\begin{proof}
Let $t_1,t_2\in T$ be arbitrary two terminal vertices and let \mbox{$Q = (t_1=t^1, t^2, \ldots, t^\ell=t_2)$} be a shortest path from $t_1$ to $t_2$ in $\sym(R)$.
Now let \mbox{$\QQQ = (Q_1,\ldots, Q_{\ell-1})$} be a realization of the path $Q$ in $H$, that is, $Q_i$ is a directed $T$-avoiding path between $t^i$ and $t^{i+1}$ of length at most $c\cdot |I_{Q_i}|$ in $H$ for every $1\le i \le \ell-1$.
Note that it does not matter whether $Q_i$ is a directed path from $t^i$ to $t^{i+1}$ or vice versa.

For $1\le i\le \ell-1$ let $g_i$ be the function $g_{Q_i}$ for the path $Q_i$ from  Lemma~\ref{lem:important_vertices_on_P}. Let $v\in I_{Q_i}$ be an important vertex on $Q_i$.
From Lemma~\ref{lem:important_vertices_on_P} it follows that there is a $(V(Q_i)\cup T)$-avoiding directed path either from $v$ to $g_i(v)$ or from $g_i(v)$ to $v$.
Moreover, since $Q_i$ is $T$-avoiding, there are two $T$-avoiding directed paths in $H$ that go either from $t^i$ to $v$ and from $v$ to $t^{i+1}$ or from $t^{i+1}$ to $v$ and from  $v$ to $t^{i}$.
Therefore, it follows from Observation~\ref{obs:adjustingR} and the choice of $R$ that if a terminal $t'$ is in $g_i(I_{Q_i})$, then there is a $T$-avoiding directed path either between $t'$ and $t^i$ or between $t'$ and $t^{i+1}$ in $H$ and consequently, by our assumptions on $R$, there is an arc between $t'$ and either $t^{i}$ or $t^{i+1}$ in $R$. 

Now, for a terminal $t'$, let $1\le i<j\le \ell-1$ be such that $t'\in \left(g_i(I_{Q_i})\cap g_j(I_{Q_j})\right)$.
Then we claim that $j-i\le 3$.
From the argument above, it follows that there is an edge between $t'$ and $t^i$ or $t^{i+1}$ and between $t'$ and $t^j$ or $t^{j+1}$ in $\sym(R)$.
However, if $j-i\ge 4$, then we can obtain a shorter path than $Q$ in $\sym(R)$ from $t_1$ to $t_2$ by going along $Q$ from $t_1$ to $t^{i}$ or to $t^{i+1}$, then using the aforementioned edges to $t'$ and from $t'$ to $t^{j}$ or $t^{j+1}$ and continuing on $Q$.
This is a contradiction with the choice of $Q$.
Therefore, for each terminal $t'$ there are at most $4$ paths $\bar{Q}\in \QQQ$ such that $t'\in g_{\bar{Q}}(I_{\bar{Q}})$.
Since for each path $\bar{Q}$ and terminal $t'$, it holds that $\left|g_{\bar{Q}}^{-1}(t') \right|\le 2$, it follows that $\sum_{i=1}^{\ell-1}|I_{Q_i}|\le 2\cdot 4\cdot |T|$.
Therefore the distance between $t_1$ and $t_2$ is at most $\sum_{i=1}^{\ell-1}|{Q_i}|\le \sum_{i=1}^{\ell-1}c\cdot|I_{Q_i}|\le 8cq$ and the lemma follows.	
\end{proof}

\begin{lemma}\appmark\label{lem:indegree_of_non_important_vertices}
	Let $p_i$, $p_j$, $p_k$ be three vertices on $P$ such that
		(1)
		$i< j < k$,
		(2)
		there is a path $Q$ from  $p_k$ to $p_i$ that avoids $p_j$, and 
		(3)
		every directed path $P'$ from some terminal $s'$ to $p_j$ in $H$ intersect $P$ in a vertex $p_\ell$ such that $p_\ell \neq p_j$ and $\ell \le k$. 
	Then $p_j$ has no in-neighbor other than $p_{j-1}$ in $H$.
\end{lemma}
\begin{proof}
Refer to Fig.~\ref{fig:indegreeOfNonimportant}.
Let $u\neq p_{j-1}$ be an in-neighbor of $p_j$. Let $s't'$ be an arc in $R$ such that the arc $up_j$ is on a path $P'$ from $s'$ to $t'$ in $H$. We show that there is a directed path from $s'$ to $t'$ in $H-up_j$. 
By our assumption, it follows that $s'[P']p_j$ intersects $P$ in a vertex $p_\ell$ such that $\ell<k$.
Therefore, the walk
$s'[P']p_\ell \circ p_\ell[P]p_k \circ p_k[Q]p_i \circ p_i[P]p_j \circ p_j[P']t'$ induces a directed path from $s'$ to $t'$ in $H-up_j$. Since this is true for every pair of terminals $s', t'$ with an $s'$-$t'$ path in $H$, it contradicts the inclusion-minimality of $H$ and hence the only in-neighbor of $p_j$ is $p_{j-1}$.
\begin{figure}[bt]
  \begin{minipage}{.5\textwidth}
  \usetikzlibrary{calc}

\begin{tikzpicture}
\tikzstyle{vertex}=[fill, circle, inner sep=2pt]
\tikzstyle{notImpVertex}=[circle, inner sep=2.2pt, fill=red]
\tikzstyle{mrkVertex}=[circle, draw, inner sep=3pt, fill=green]
\tikzstyle{mrkImpVertex}=[circle, draw=red, thick, inner sep=3pt, fill=green]

\tikzstyle{blackPath}=[->, thick]
\tikzstyle{greenPath}=[blackPath, green!40]
\tikzstyle{redPath}=[blackPath, red!90]


\node[label={45:$P$}] (P0) {};
\node[] at ($(P0) + (6.5,0)$) (PLast) {};
\draw[blackPath] (P0) to (PLast);

\node[label={270:$p_i$},vertex] at ($(P0)!.2!(PLast)$) (pi) {};
\node[label={270:$p_j$},vertex] at ($(P0)!.5!(PLast)$) (pj) {};
\node[label={270:$p_k$},vertex] at ($(P0)!.8!(PLast)$) (pk) {};

\draw[draw=red, ->, thick, dashed] (pk) to[out=135, in=45] node[above, xshift=1.5cm, yshift=-.5cm] {Q} (pi);

\coordinate (pl) at ($(P0)!.7!(PLast)$);
\draw[draw=orange!60,thick,->] ($(pl) - (.3,.3)$) node[below,yshift=.18cm,xshift=.05cm] {$s'$} to (pl) to[out=135,in=45] (pj) to ($(pj) - (1, .3)$) node[below,yshift=.2cm,xshift=-.2cm] {$t'$};
\end{tikzpicture}
  \end{minipage}
  \begin{minipage}{.47\textwidth}
  \caption{\label{fig:indegreeOfNonimportant}
  The three vertices $p_i,p_j,p_k$ on a directed path $P$ as in Lemma~\ref{lem:indegree_of_non_important_vertices}.
  The orange (light gray) path cannot exists as it is rerouted via $p_k$ and $Q$ (dashed); contradicting the minimality of the solution.
  }
  \end{minipage}
\end{figure}
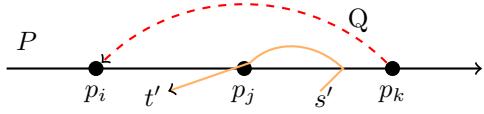
\end{proof}

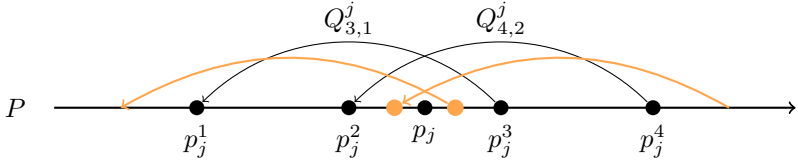
\begin{figure}
  \usetikzlibrary{calc}

\begin{tikzpicture}
\tikzstyle{vertex}=[fill, circle, inner sep=2pt]
\tikzstyle{notImpVertex}=[circle, inner sep=2.2pt, fill=orange!70]
\tikzstyle{mrkVertex}=[circle, draw, inner sep=3pt, fill=green]
\tikzstyle{mrkImpVertex}=[circle, draw=red, thick, inner sep=3pt, fill=green]

\tikzstyle{blackPath}=[->, thick]
\tikzstyle{greenPath}=[blackPath, green!40]
\tikzstyle{redPath}=[blackPath, red!90]


\node[label={180:$P$}] (P0) {};
\node[] at ($(P0) + (10,0)$) (PLast) {};
\draw[blackPath] (P0) to (PLast);

\node[label={270:$p_j^1$},vertex] at ($(P0)!.2!(PLast)$) (pj1) {};
\node[notImpVertex] at ($(P0)!.54!(PLast)$) (first) {};
\node[label={270:$p_j^2$},vertex] at ($(P0)!.4!(PLast)$) (pj2) {};
\node[label={270:$p_j$},vertex] at ($(P0)!.5!(PLast)$) (pj) {};
\node[label={270:$p_j^3$},vertex] at ($(P0)!.6!(PLast)$) (pj3) {};
\node[notImpVertex] at ($(P0)!.46!(PLast)$) (second) {};
\node[label={270:$p_j^4$},vertex] at ($(P0)!.8!(PLast)$) (pj4) {};


\draw[->] (pj4) to[out=135, in=45] node[midway, yshift=.3cm] {$Q^j_{4,2}$} (pj2);
\draw[->] (pj3) to[out=135, in=45] node[midway, yshift=.3cm] {$Q^j_{3,1}$}(pj1);

\draw[->,orange!70,thick] (first) to[out=150, in=30] ($(P0)!.1!(PLast)$);
\draw[->,orange!70,thick] ($(P0)!.9!(PLast)$) to[out=150, in=30] (second);

\end{tikzpicture}
  \caption{\label{fig:atMostTwoPoints}
  The four vertices on path $P$ for the vertex $p_j$.
  By the choice of $p_j^1$ and $p_j^4$, the orange (light gray) paths cannot exist.
  }
\end{figure}
For a vertex $p_j\in V(P)$ let $p_j^1, p_j^2,p_j^3,p_j^4$ denote the following four vertices (see Fig.~\ref{fig:atMostTwoPoints}):
\begin{itemize}
	\item $p_j^1$ is the leftmost vertex on $P$ such that there is a $P$-avoiding path from a vertex $p_x$, with $x\ge j$ to $p_j^1$,
	\item $p_j^3$  is a vertex such that $p_j\le_P p_j^3$ and $p_j^3$ is the first vertex of some $P$-avoiding path to $p_j^1$,
	\item $p_j^4$ is the rightmost vertex on $P$ such that there is a $P$-avoiding path from $p_j^4$, to some vertex $p_y$ with $y \le j$, and
	\item $p_j^2$, is a vertex such that  $p_j^2\le_P p_j$ is the last vertex of some $P$-avoiding path from $p_j^4$.
\end{itemize}
\lv{
Note that $p_j^1\le_P p_j^2\le_P p_j \le_P p_j^3 \le_P p_j^4$ and that some of the vertices $\left\{p_j, p_j^1, p_j^2, p_j^3, p_j^4 \right\}$ might coincide.
In particular, if there is no $P$-avoiding path that starts in or after $p_j$ on $P$ and ends in or before $p_j$, then $\left\{p_j, p_j^1, p_j^2, p_j^3, p_j^4 \right\}=\left\{p_j\right\}$.
}
Furthermore, let us denote $Q_{3,1}^j$ and $Q_{4,2}^j$ the $P$-avoiding paths from $p_j^3$ to $p_j^1$ and from $p_j^4$ to $p_j^2$, respectively, and let $Q_{4,1}^j$ denote the path $Q_{4,2}^j\circ p_j^2[P]p_j^3\circ Q^j_{3,1}$.

\begin{lemma}\appmark\label{lem:vertices_around_the_I_P}
For every $p_j$ there are at most two vertices in $V(P)\setminus (I_P\cup \{p_j^2, p_j^3\})$ between $p_j^1$ and $p_j^4$. 
\end{lemma}
\toappendix{
\begin{proof}[Proof of Lemma~\ref{lem:vertices_around_the_I_P}]
We first show that all the non-important vertices with respect to $P$ between $p_j^1$ and $p_j$ other than  $p_j^2$ have in-degree one.

Let $p_x$ be a non-important vertex between $p_j^1$ and $p_j$ on $P$. Since $p_x$ is not important, every path $P'$ from a terminal $s'$ to $p_x$ intersects $P$. The subpath of $P'$ between the last intersection of $P'$ with $P$, say $p_y$, and $p_x$ is $P$-avoiding. Hence $p_y \le_P p_j^4$ as $p_x \le_P p_j$. Because the choice of $P'$ was arbitrary, every path from a terminal to $p_x$ intersects $P$ is some vertex $p_y$ with $p_y\le_P p_j^4$. If $p_j^2\le_P p_x$, then the vertices  $p_j^2$, $p_x$, $p_j^4$ satisfy all the assumptions of Lemma~\ref{lem:indegree_of_non_important_vertices}. Otherwise, $p_j^1\le_P p_x\le_P p_j^2$, $Q_{4,1}^j$ avoids $p_x$, and the vertices  $p_j^1$, $p_x$, $p_j^4$ satisfy all the assumptions of Lemma~\ref{lem:indegree_of_non_important_vertices}. 
Hence, $p_x$ has only one in-neighbor, namely $p_{x-1}$.

It follows that every non-important vertex $v$ between $p^1_j$ and $p_j$ has out-degree at least two.
By Lemma~\ref{lem:nonimportant} for every non-important vertex $v$ with respect to $P$ with out-degree at least two there is a $P$-avoiding path that starts in this vertex and ends in a different vertex on $P$.
Let $p_\ell$ be the vertex in $V(P)\setminus (I_P\cup \{p_j^2, p_j^3\})$ between $p_j^1$ and $p_j$ with out-degree at least two such that a $P$-avoiding path starting in $p_\ell$ takes us to a vertex $p_{\ell'}$ with the lowest index $\ell'$.
Let $Q_\ell$ be a $P$-avoiding path from $p_\ell$ to $p_{\ell'}$. We claim that $p_\ell$ is the only vertex in $V(P)\setminus (I_P\cup \{p_j^2, p_j^3\})$ between $p_j^1$ and $p_j$.
Let $p_k$ be a non-important vertex between $p_j^1$ and $p_j$ other than $p_\ell$. Note that from the choice of $p_\ell$ it follows that every $P$-avoiding path from $p_k$ back to $P$ ends in some vertex $p_{k'}$ with $\ell'\le k'$. 
Now if $k\le \ell$ then Lemma~\ref{lem:indegree_of_non_important_vertices} applies to $p_{\ell'}$, $p_{k'}$, and $p_{\ell}$ contradicting the existence of a $P$-avoiding path from $p_{k}$ to $p_{k'}$.
Otherwise $k > \ell$ and the path $Q_{3,1}^j\circ p_j^1[P]p_\ell \circ Q_\ell$ avoids $p_k$ and the Lemma~\ref{lem:indegree_of_non_important_vertices} applied to one of $p_\ell$ or $p_j^3$, together with $p_k$, and $p_{\ell'}$ in $\overleftarrow{H}$ implies that $p_k$ has out-degree one in $H$.
Hence, $p_k$ has total degree at most two -- a contradiction with the choice of $H$ (see Lemma~\ref{lem:degreeTwoVertices}).

Therefore, there is at most one vertex in $V(P)\setminus (I_P\cup \{p_j^2, p_j^3\})$ between $p_j^1$ and $p_j$, namely $p_\ell$.
Using the same arguments in $\overleftarrow{H}$, we obtain that there is at most one vertex in $V(P)\setminus (I_P\cup \{p_j^2, p_j^3\})$ between $p_j$ and $p_j^4$.
\end{proof}
}

For the rest of this section, let us define 
the set $Q_P=\left\{p_j^1, p_j^2, p_j^3, p_j^4\mid p_j\in I_P \right\}$. It is easy to see that $\left|Q_P\right|\le 4\left|I_P\right|$. We will call the set $Q_P$ the \emph{set of marked vertices for $P$}.

Note that the same vertex in $Q_P$ may be marked for different reasons at the same time.
That is, for example, the same vertex can be denoted $p_j^1$, because it is the first vertex for the important vertex $p_j$ and at the same time it can be denoted $p_k^3$, because it is also third marked vertex for the important vertex $p_k$ with respect to $P$.

\subsection{Ladders}
In this subsection we define ladder graphs.
These graphs play crucial role as we will be able to show that if there is a $T$-avoiding $s$-$t$-path for $st\in A(R)$ that is ``long'', then in $H$ there is a ``large'' ladder (Lemma~\ref{lem:ladder_structure_between_important_points}).
Moreover, it is possible to replace such a ladder with one having constant size while preserving all connections and inclusion-minimality (Lemma~\ref{lem:protrusion_replacement}).

\begin{definition}[Class of Ladders]
Let $n$ be a positive integer and $I \subseteq [n]$ a set.
For examples of ladder graphs see Fig.~\ref{fig:ladder}.
We define the directed graph $G_{n}$ and the directed graph $G_{n, I}$ as follows.
Vertex set $V(G_{n})$ is the set $\left\{ a_i, b_i \mid i \in [n]\right\}$ the arc set $A(G_{n})$ is the set 
$
\left\{ a_{2i+1}b_{2i+1} \mid 0 \le i < n/2 \right\} \cup	
\left\{ b_{2i}a_{2i} \mid  1\le i \le n/2 \right\} \cup		
\left\{ a_{2i}a_{2i-1} \mid 1 \le i \le n/2 \right\} \cup 
\left\{ a_{2i}a_{2i+1} \mid 1 \le i < n/2 \right\} \cup	
\left\{ b_{2i+1}b_{2i} \mid 1 \le i < n/2 \right\} \cup	
\left\{ b_{2i-1}b_{2i} \mid 1 \le i \leq n/2 \right\} 
$.
The graph $G_{n, I}$ is the graph $G_{n}$ where we identify the vertices $a_i$ and $b_i$ whenever $i \in I$ (i.e., $G_{n}$ and $G_{n, \emptyset}$ is the same graph).
We emphasize that we suppress any loops in $G_{n,I}$.
We say that $n$ is the \emph{length} of the ladder $G_{n,I}$.
\end{definition}
\begin{figure}[bt]
  \usetikzlibrary{calc}

\begin{tikzpicture}
\tikzstyle{vertex}=[circle, inner sep=1pt, draw]
\tikzstyle{Larrow}=[->,thick]
\tikzstyle{Rarrow}=[<-,thick]
\tikzstyle{LarrowShortened}=[->,thick,shorten >=14pt]
\tikzstyle{RarrowShortened}=[<-,thick,shorten >=15pt]

\newcommand{\NUM}{6}
\pgfmathsetmacro\NUMm{\NUM - 2}

\foreach \x in {1,2,...,\NUM} {
  \node[vertex] (a\x) at (\x, 1) {$a_{\x}$};
  \node[vertex] (b\x) at (\x, 0) {$b_{\x}$};
}


\foreach \x in {1,3,...,\NUM} {
  \draw[Larrow] (a\x) -- (b\x);
}

\foreach \x in {2,4,...,\NUM} {
  \draw[Larrow] (b\x) -- (a\x);
}

\foreach \x in {2, 4, ..., \NUM} {
  \pgfmathsetmacro\xm{\x -1}
  \draw[Larrow] (a\x) -- (a\xm);
}

\foreach \x in {2, 4, ..., \NUMm} {
  \pgfmathsetmacro\xp{\x + 1}
  \draw[LarrowShortened] (a\x) -- (a\xp);
}

\foreach \x in {2, 4, ..., \NUM} {
  \pgfmathsetmacro\xm{\x -1}
  \draw[Rarrow] (b\x) -- (b\xm);
}

\foreach \x in {2, 4, ..., \NUMm} {
  \pgfmathsetmacro\xp{\x + 1}
  \draw[RarrowShortened] (b\x) -- (b\xp);
}

\begin{scope}[xshift=7cm]
\foreach \x in {1,3,4,6} {
  \node[vertex] (a\x) at (\x, 1) {$a_{\x}$};
  \node[vertex] (b\x) at (\x, 0) {$b_{\x}$};
}

\node[vertex] (a2) at (2, .5) {$i_{2}$};
\node[vertex] (a5) at (5, .5) {$i_{5}$};

\draw[Larrow] (a1) -- (b1);
\draw[Larrow] (a3) -- (b3);
\draw[Larrow] (a2) -- (a1);
\draw[Larrow] (b1) -- (a2);
\draw[Larrow] (a2) -- (a3);

\draw[Larrow] (b3) -- (a2);
\draw[Larrow] (b4) -- (a4);
\draw[Larrow] (a4) -- (a3);
\draw[Larrow] (b3) -- (b4);

\draw[Larrow] (a4) -- (a5);
\draw[Larrow] (a5) -- (b4);
\draw[Larrow] (a6) -- (a5);
\draw[Larrow] (a5) -- (b6);

\draw[Larrow] (b6) -- (a6);
\draw[Larrow] (a1) -- (b1);

\end{scope}

\end{tikzpicture}
  \caption{\label{fig:ladder}%
  Examples ladder graphs.
  The ladder $G_{6} = G_{6, \emptyset}$ to the left and $G_{6, \{2,5\}}$ to the right.
  }
\end{figure}
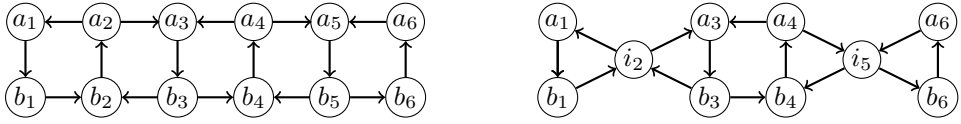

\begin{lemma}\appmark\label{lem:ladder_properties}
	Given a positive integer $n$ and $I\subseteq [n]$, the ladder $G_{n,I}$ is a union of two paths $P_1$ from $a_1$ to $a_n$  and $P_2$ from $b_n$ to $b_1$ if $n$ is even or paths $P_1$ from $a_1$ to $b_n$  and $P_2$ from $a_n$ to $b_1$, if $n$ is odd. Moreover, $G_{n,I}$ is an inclusion-minimal strongly connected graph connecting the set of terminals $\{a_1,b_1,a_n,b_n\}$. 	
\end{lemma}
\toappendix{
\begin{proof}[Proof of Lemma~\ref{lem:ladder_properties}]
Let $P_1$ be the path containing precisely the arcs 
\begin{align*}
\left\{ a_{2i+1}b_{2i+1} \mid 0 \le i < n/2 \right\} \cup	
\left\{ b_{2i}a_{2i} \mid  1\le i \le n/2 \right\} \cup		
&\left\{ a_{2i}a_{2i+1} \mid 1 \le i < n/2 \right\} \cup\\	
&\left\{ b_{2i-1}b_{2i} \mid 1 \le i \le n/2 \right\}.	
\end{align*}
and let $P_2$ be the path containing precisely the arcs 
\begin{align*}
\left\{ a_{2i+1}b_{2i+1} \mid 0 \le i < n/2 \right\} \cup	
\left\{ b_{2i}a_{2i} \mid  1\le i \le n/2 \right\} \cup		
&\left\{ a_{2i}a_{2i-1} \mid 1 \le i \le n/2 \right\} \cup\\	
&\left\{ b_{2i+1}b_{2i} \mid 1 \le i < n/2 \right\}.	
\end{align*}
Clearly, the union of $P_1$ and $P_2$ contains all the arcs in $G_{n,I}$.
Furthermore, the inclusion-minimality is straightforward to verify (refer to Fig.~\ref{fig:ladder}).
\end{proof}
}

\toappendix{
\begin{lemma}[Sufficient Condition]\appmark\label{lem:ladder_properties_only_if}
	Let $K$ be a directed graph, with $a,b,c,d\in V(K)$ such that $a=b$ or $ab\in A(K)$, $c=d$ or $cd\in A(K)$, and $K$ is an inclusion-minimal graph such that 
	\begin{itemize}
	\item if $a\neq b$, then  $a$ is the only in-neighbor of $b$ and $b$ is the only out-neighbor of $a$,
	\item if $c\neq d$, then and $c$ is the only in-neighbor of $d$ and $d$ is the only out-neighbor of $c$,
	\item  there exists a directed path $P_1$ from $a$ to $d$, and
	\item  there exists a directed path $P_2$ from $c$ to $b$.
	\end{itemize}
	Then $K$ is a subdivision of a ladder. 
\end{lemma}
}
\toappendix{
\begin{proof}
	We prove the lemma by induction on the number of vertices. It is easy to verify that it holds if $|V(K)|\le 4$. 
	
	If $K$ contains a vertex $v$ with total degree at most $2$ other than one of $\{a,b,c,d\}$, then since $K$ is inclusion-minimal, it follows that $v$ is not a sink or source, but has one in- and one out-neighbor. Replacing $v$ by an arc between its neighbors does not falsify any of the assumptions. Hence by the induction hypothesis, after replacing $v$ by an arc between the neighbors, the graph is a subdivision of the ladder. Hence also $K$ is a subdivision of the ladder. 
	
	For the rest of the proof, we assume that there is no vertex of total degree $2$ in $K - \{a,b,c,d\}$. Now, note that $K$ contains precisely the arcs in $A(P_1)\cup A(P_2)$. Furthermore, as there is no vertex of total degree $2$ in $K-\{a,b,c,d\}$, and arcs $ab$ and $cd$ are necessarily on both $P_1$ and $P_2$, it follows that every vertex of $K$ lies on both $P_1$ and $P_2$. 	

Let $\bar{a}$ be the unique in-neighbor of $a$ on $P_2$ and let $\bar{b}$ be the unique out-neighbor of $b$ on $P_1$. Note that $a$ and $b$ does not have any other neighbors, as arc $ab$ is on both $P_1$ and $P_2$.
Clearly, if $\bar{a}=\bar{b}$, then by induction hypothesis applied to $K-\{a,b\}$ and $\bar{b},\bar{a},c,d$ we get that $K-\{a,b\}$ is a ladder.
Now it follows that $K$ is also a ladder, since $\bar{a},\bar{b}$ are the only neighbors of $a,b$ in the graph $K-\{a,b\}$.
	
	We claim that if $\bar{a}\neq \bar{b}$, then $\bar{b}$ is the only in-neighbor of $\bar{a}$ and $\bar{a}$ is the only out-neighbor of $\bar{b}$ in $K-\{a,b\}$. 
	Note that both $P_1$ and $P_2$ contain subpath from $\bar{b}$ to $\bar{a}$. However $a[P_1]\bar{b}\circ \bar{b}[P_2]\bar{a}\circ \bar{a}[P_1]d$ is also an $a$-$d$ path and $c[P_2]\bar{b}\circ \bar{b}[P_1]\bar{a}\circ \bar{a}[P_2]b$ is a $c$-$b$ path. Hence, either $\bar{b}[P_2]\bar{a} = \bar{b}[P_1]\bar{a}$, or we can remove an arc from one of the paths. However, in the second case, $K$ is not inclusion-minimal, which is a contradiction. In the first case all inner vertices on the path $\bar{b}[P_2]\bar{a}$ have the same in- and the same out-neighbor on both $P_1$ and $P_2$ and hence they have total degree $2$. Since there are no such vertices, it follow that $\bar{b}$ is the only in-neighbor of $\bar{a}$ and $\bar{a}$ is the only out-neighbor of $\bar{b}$.
Then the graph $K-\{a,b\}$ together with $\bar{b},\bar{a},c,d$ satisfy all the assumptions of the lemma. It follows from the induction hypothesis and the fact that $\bar{a},\bar{b}$ are the only neighbors of $a,b$ in $V(K) \setminus \{a,b\}$ that $K$ is also a ladder.
\end{proof}
}

\toappendix{
\begin{lemma}\appmark\label{lem:outerplanar_ladder}
Let $K$ be a $2$-connected outerplanar graph and let $a,b,c,d$ be vertices of degree two such that $ab, cd \in E(K)$.
If $\deg(v) = 3$ for all $v\in V(K)\setminus \{a,b,c,d\}$, then $K$ is the underlying undirected graph of a ladder.
\end{lemma}
}
\toappendix{
\begin{proof}
Be begin with a technical claim.
\begin{claim}\label{cml:degree_2_outerplanar}
Let $\tilde{K}$ be a maximal outerplanar graph and let $u,v$ be degree two vertices with $uv\in E(\tilde{K})$.
Then there is a vertex $w \in V(\tilde{K})\setminus \{u,v\}$ with degree $2$.
\end{claim}
\begin{claimproof}
Since $\tilde{K}$ is maximal outerplanar, it follows that it can be constructed from a triangle by gluing triangles. 
As $u$ and $v$ have degree two there exists their common neighbor $x$.
Now we may assume that the construction started from the triangle with vertices $u,v,x$.
Let $w$ be the vertex ``introduced'' by the last gluing operation -- it follows that its degree is $2$.
\end{claimproof}

The prove is done by induction on $|K|$.
If $|K| = 2$ then $K$ is $C_4$ and thus an underlying undirected graph of a ladder.

For the induction step let $\bar{a}$ be the neighbor of $a$ different from $b$ and similarly $\bar{b}$ for $b$.
We would like to show that $\bar{a}\bar{b} \in E(K)$.
If this is so the lemma follows from induction hypothesis applied to $K - \{a,b\}$.

Let $a,b,c,d$ be the order in which these vertices appear on the outer face of $K$ and suppose $\bar{a}\bar{b} \notin E(K)$.
Now one neighbor of $\bar{b}$ is next to it in the fixed ordering along the outer face.
We show that the other neighbor $b'$ of $\bar{b}$ must lie between $d$ and $\bar{a}$.
Suppose that this is not the case then the edge $\bar{b}b'$ encloses an outerplanar graph $K'$ a subgraph of $K$ not containing any of $a,b,c,d$.
If we apply Claim~\ref{cml:degree_2_outerplanar} to an outerplanar triangulation of $K'$ and $\bar{b},b'$, we obtain a vertex of degree (at most) two in $K'$.
This is impossible as all vertices in $K'$ must have degree three -- a contradiction.
Similarly it follows that $\bar{a}$ has the other neighbor (non-consecutive in the fixed order along the outer face of $K$) between $b$ and $c$.
It follows that $\bar{a}\bar{b}$ is an edge, since this is the only possibility as $K$ is outerplanar.
\end{proof}
}

\subsection{Finishing the Proof}
The following technical lemma we show that if the distance between any two consecutive vertices $p_i, p_j \in Q_P \cup I_P$ with $i<j$ is at least $5$, then $p_i = p_k^4$ and $p_j = p_\ell^1$ where $p_k,p_\ell \in I_P$ and $k<\ell$.
Moreover, there exists a path from $p_j$ to $p_i$ in $H$ and between $p_i$ and $p_j$ there is a ladder with constant-sized boundary.

\begin{lemma}\appmark\label{lem:ladder_structure_between_important_points}
	Let $p_i, p_j \in Q_P \cup I_P$ with $i<j$ such that there is no $p\in Q_P\cup I_P$ with $p_i\le_P p \le_P p_j$.
	Let $F= \{p_{i+1},\ldots, p_{j-1}\}$ and let $C$ be the connected component of $H-\{p_{i+1},p_{i+2}, p_{j-2}, p_{j-1}\}$ containing $p_{i+3}$.
	If $j-i\ge 5$, then $H[C\cup \{p_{i+1},p_{i+2}, p_{j-2}, p_{j-1}\}]$ is a ladder and furthermore, $p_{i+1},p_{i+2}, p_{j-2}$, and $p_{j-1}$ are the only vertices with a neighbor outside the ladder.
\end{lemma}
\toappendix{
\begin{proof}[Proof of Lemma~\ref{lem:ladder_structure_between_important_points}]
%
%
%
We distinguish cases depending on whether $p_i$ and $p_j$ are important vertices or first, second, third, or forth marked vertex for some important vertex on $P$.  


	\begin{description}
		\item[Case 1)]  $p_i$ or $p_j$ is important. 
	\end{description}

	Without loss of generality, let us assume that $p_j$ is important. The other case is symmetric and follows from following the same argument in $\overleftarrow{H}$.

If $p_j\neq p_j^1$, then clearly $p_j^1\le_P p_i \le_P p_j$, because there is no marked vertex between $p_i$ and $p_j$ and $p_j^1$ is a marked vertex before $p_j$. Lemma~\ref{lem:vertices_around_the_I_P} says that there are at most $2$ non-important vertices between $p_j^1$ and $p_j$. However, then $j-i\le 3$ and the case follows. 

On the other hand, if $p_j=p_j^1$, then we can again show that there is at most one vertex of total degree $3$ between $p_i$ and $p_j$.

First, suppose that a vertex $p_k$ between $p_i$ and $p_j$ has in-degree $2$. 
It follows from the definition of (non-)important vertices that every path from some terminal to $p_k$ has to intersect $P$ in an important vertex, since $p_k$ is not important.
However, this vertex has to be before $p_j$, otherwise there would be a $P$-avoiding path from a vertex after $p_j$ to a vertex before $p_j$ contradicting the assumption $p_j = p_j^1$. Since there are no important vertices between $p_i$ and $p_j$, every path $P'$ from a terminal $s'$ to $p_k$ intersect $P$ in an important vertex $p_{k'}\le_P p_i$.
Clearly, every such path can use the subpath of $P$ between $p_{k'}$ and $p_k$.
Therefore, it follows from the inclusion-minimality of $H$ that $p_k$ has in-degree $1$. 

Second, we show that there is at most one vertex with out-degree at least two between $p_i$ and $p_j$. 
Now let $p_k$ be a vertex between $p_i$ and $p_j$ other than $p_i^3$. As argued before, every path $P'$ from a terminal $s'$ to $p_k$ intersect $P$ in an important vertex $p_{k'}\le_P p_i$. Furthermore, since $p_k$ is not important, every path $P''$ from $p_k$ to some terminal $t'$ intersect $P$ in some vertex $p_{k''}$ such that $p_i^1\le_P p_{k''}$. However, it is easy to see that every path from $s'$ to $t'$ can use the path $p_{k'}[P]p_i^3\circ Q_{3,1}^i\circ p_i^1[P]p_{k''}$ between $p_{k'}$ and $p_{k''}$ and hence $p_k$ has out-degree at most one as well. 
It follows that there is at most one vertex of total degree at least $3$ between $p_i$ and $p_j$. Since $H$ does not contain vertices of total degree at most $2$, it follows that $j-i\le 3$. 

%


We are now left only with the cases, when both $p_i$ and $p_j$ are marked and not important. 

\begin{description}
	\item[Case 2)]
	there exists $k$ such that $p_k^1\le p_i < p_j \le p_k^4$
\end{description}

By Lemma~\ref{lem:vertices_around_the_I_P} there are at most two non-important vertices between $p_k^1$ and $p_k^4$ and it follows that $j-i\le 3$. 

Note that if $p_i=p_k^a$, $a\in[4]$, but $p_j$ is not between $p_k^1$ and $p_k^4$, then $p_k^4 < p_j$ and necessarily $p_i=p_k^4$. Similarly, if $p_j=p_\ell^b$, $b\in [4]$, but $p_i$ is not between $p_\ell^1$ and $p_\ell^4$, then $p_i < p_\ell^1$ and necessarily $p_j=p_\ell^1$. Hence, we are left with the following case:

	\begin{description}
	\item[Case 3)]
	$p_i$ and $p_j$ are $p_k^4$ and $p_{\ell}^1$ for some important vertices $p_k$, $p_\ell$ with $k < \ell$ and Case 2) does not apply.
	\end{description}
%
	Let us first focus on the vertices between $p_i$ and $p_j$ that are last vertices of some $P$-avoiding path that starts on a vertex in $V(P)\setminus F$. Clearly $p_{i+1}$ is such a vertex, as it has an arc from $p_i$. We will show that there is at most one more such vertex and that this vertex is either $p_{j-1}$ or $p_{j-2}$. Symmetrically, we can reach some vertex in $V(P)\setminus F$ by a $P$-avoiding path only from one of the vertices $p_{j-1}, p_{i+1}$, and $p_{i+2}$.
	
	 Since all vertices between $p_i$ and $p_j$ are not important, every path from some terminal $s'$ to a vertex $p_x$ between $p_i$ and $p_j$ intersect $P$ in some (important) vertex $p_y$. If $y < x$, then each such path can use path $y[P]x$ between $y$ and $x$ instead. Hence it follows that $y>x$ and in particular any $P$-avoiding path that starts on a vertex in $V(P)\setminus F$ and ends in $p_x$ has to start in a vertex that is after $p_j$ on $P$.  
	
	 Now let $p_x$ be the leftmost vertex on $P$ 
	 such that $p_x$ is an endpoint of a $P$-avoiding path $Q$ that starts in a vertex $p_y$ on $P$ with $y\ge j$. Note that $p_\ell^4[Q_2^\ell]p_\ell^2\circ p_\ell^2[P]p_\ell^3\circ p_\ell^3[Q_1^\ell]p_\ell^1\circ p_\ell^1[P]p_y\circ p_y[Q]p_x$ is a directed path between $p_\ell^4$ and $p_x$ avoiding all the vertices between $p_x$ and $p_j=p_\ell^1$. Moreover, the vertices between $p_x$ and $p_j$ are all not important. Hence, it follows from the choice of $p_\ell^4$ that every directed path $P'$ from a terminal $s'$ to a vertex between $p_x$ and $p_j$ intersect $P$ in some vertex $p_z$ such that $p_z\le_P p_\ell^4$. Therefore, Lemma~\ref{lem:indegree_of_non_important_vertices} applied to $p_x$, any vertex between $p_x$ and $p_j$, and $p_\ell^4$ implies that $p_x$ and $p_{i+1}$ are only two vertices between $p_i$ and $p_j$ on $P$ with a in-neighbor in $H-F$. And in particular all the vertices between $p_x$ and $p_{j}$ have in-degree $1$ in $H$. 
	 Now we show that there is at most one vertex between $p_x$ and $p_j$. As all of these vertices are non-important with in-degree one, they have out-degree two. Furthermore, by Lemma~\ref{lem:nonimportant} there is a $P$-avoiding path from each such vertex $p_{x'}$ to a vertex $p_{y'}$ with $y'<x'$. But all the vertices between $p_x$ and $p_j$ have in-degree one and, hence, also $y'\le x$. Thus, there is a $P$-avoiding path from $p_{j-1}$ to some vertex $p_z\le_P p_x$. Applying Lemma~\ref{lem:indegree_of_non_important_vertices} to $p_{j-1}, p_{x'}, p_z$ for every $x'$ between $x$ and $j-1$ in $\overleftarrow{H}$ implies that all vertices between $x$ and $j-1$ have out-degree one. But we already proved that they have also in-degree two, which contradicts the fact that every vertex in $H$ has at least three neighbors. 
	 By a symmetric argument in  $\overleftarrow{G}$, $\overleftarrow{H}$, and $\overleftarrow{R}$, we obtain that only the vertices $p_{i+1}$ and $p_{i+2}$ (and obviously $p_{j-1}$) can have an $F$-avoiding directed path to a vertex in $V(P)\setminus F$.
	 
	 Now clearly the only role of edges in $H[C\cup F]$ is to connect vertices on $F$ that are the endpoints of some $F$-avoiding directed path to or from a vertex in $V(P)\setminus F$. We know that only the vertex $p_{i+1}$ and one of vertices $p_{j-2}$ or $p_{j-1}$ can have an $F$-avoiding path from a vertex in $V(P)\setminus F$. Similarly, only the vertex $p_{j-1}$ and one of vertices $p_{i+1}$ or $p_{i+2}$ can have an $F$-avoiding path to a vertex in $V(P)\setminus F$. Let $u$ be the (unique) vertex other than $p_{j-1}$ that is the first vertex on an $F$-avoiding path to a vertex in $V(P)\setminus F$ and let $v$ be the (unique) vertex other than $p_{i+1}$ that is the last vertex on an $F$-avoiding path from a vertex in $V(P)\setminus F$. Now we claim that $H[C\cup F]$ together with vertices $p_{i+1}, v,u, p_{j-1}$ satisfy all the conditions of Lemma~\ref{lem:ladder_properties_only_if}
	 
	 
	  Clearly, $p_{i+1}[P]p_{j-1}$ is a path from $p_{i+1}$ to $p_{j-1}$. Furthermore, if there would be no path from $v$ to $u$ in $H[C\cup F]$, then all arcs not in $A(P)$ that are incident to a vertex between $u$ and $v$ 
	  would be pointless and, because of inclusion-minimality, not present. Therefore, all the vertices between $u$ and $v$ would have total degree at most two and it follow that $j-i\le 5$. Otherwise, $H[C\cup F]$ contains also a path from $v$ to $u$. Furthermore, either $u= p_{j-1}$, or $u= p_{j-2}$, but then  $p_j^1$ is before $p_{j-1}$ and it follows from Lemma~\ref{lem:indegree_of_non_important_vertices} applied to $p_j^1$, $p_{j-1}$, and $p_j^4$ that $p_{j-2}$ is the only in-neighbor of $p_{j-1}$ in $H[C\cup F]$. 
	  From a symmetric argument it follows that either $v= p_{i+1}$, or $v= p_{i+2}$. Furthermore, from inclusion-minimality of $H$, it follows that $H[C\cup F]$ is also inclusion-minimal. Hence, indeed $H[C\cup F]$, together with vertices $p_{i+1}, v,u, p_{j-1}$ satisfy all the assumptions of Lemma~\ref{lem:ladder_properties_only_if} and hence it is a ladder and the lemma follows. 
\end{proof}
}

\begin{lemma}\label{lem:protrusion_replacement}
	Let $a,b,c,d$ be four vertices of $H$ and $F\subseteq V(H)$, such that $a=b$ or $ab\in A(H)$, $c=d$ or $cd\in A(H)$,  $F\cap T = \emptyset$, $H[F]$ is a connected component of $H-\{a,b,c,d\}$, and $H[F\cup \{a,b,c,d\}]$ is isomorphic to a ladder $G_{n,I}$. There exist a directed graph $H'$ and a set $F'\subseteq V(H')$ such that:
	(1)
	the genus of $H'$ is at most the genus of $H$, 
	(2)
	$H'-F'=H-F$,
	(3)
	$|F'|= \OhOp{1}$,
	(4)
	$N_{H'}(F')=\{a,b,c,d\}$,
	(5)
	$H'$ is an inclusion-minimal solution to $R$, and 
	(6)
	for every $k\ge 10$, if $\sym(H)$ contains $k\times k$ grid as a minor, then $\sym(H')$ contains $k\times k$ grid as a minor. 
\end{lemma}
\sv{
\begin{proof}[Proof sketch]
	From Lemma~\ref{lem:ladder_properties} it follows that $H[F\cup \{a,b,c,d\}]$ is a union of two directed paths $P_1$ from $a$ to $d$ and $P_2$ from $c$ to $b$. We construct $F'$ such that $H'[F'\cup \{a,b,c,d\}]$ is a ladder $G_{n',\emptyset}$. Even though it is a bit technical, it is rather straightforward to verify that if we replace $F$ by another ladder, then $H'$ remain an inclusion-minimal solution to $R$. If $\sym(H)$ does not contain any $k\times k$ grid for $k\ge 10$, then we just replace $F$ with any constant size ladder and we are fine. Otherwise, we take a largest grid minor $K$ of $\sym(H)$. Since $\sym(H)[F\cup \{a,b,c,d\}]$ has treewidth $2$ and only $4$ of its vertices have neighbors in the rest of $H$, one can show that $\sym(H)[F\cup \{a,b,c,d\}]$ contracts to at most eight vertices in $K$. Let $K_F$ be the graph induced on these eight vertices. It is easy to see that if we replace $H[F\cup \{a,b,c,d\}]$ with any ladder whose underlying undirected graph has $K_F$ as a minor which furthermore maps its boundaries onto $K_F$ in a same way as $\sym(H)[F\cup \{a,b,c,d\}]$, then the underlying undirected graph of the resulting graph contains $K$ as a minor as well. However, one can express by constant size MSO formula that a boundaried graph is a ladder $G_{n',\emptyset}$ and has the boundaried graph $K_F$ as a minor. It follows that this formula has a constant size model, whose suitable orientation is the sought replacement. 
\end{proof}
}
\toappendix{
\begin{proof}[Proof of Lemma~\ref{lem:protrusion_replacement}]
	From Lemma~\ref{lem:ladder_properties} it follows that $H[F\cup \{a,b,c,d\}]$ is a union of two directed paths $P_1$ from $a$ to $d$ and $P_2$ from $c$ to $b$. 
	 We will construct $F'$ and the edges incident to $F'$ in a way such that 
	 $H'[F'\cup \{a,b,c,d\}]$ is a ladder $G_{n',\emptyset}$ and the set $\{a,b,c,d\}$ is a vertex-cut in $H'$.
	
\begin{claim}
$H'$ is an inclusion-minimal solution to $R$.
\end{claim}	
\begin{claimproof}
Observe, that $H'$ is a solution to $R$, because any $s$-$t$ path $P$ between two terminals $s$ and $t$ in $H$ either did not contain any vertex of $F$ and then it is present in $H'$ as well, or it somewhere entered $F$ for the first time from a vertex $x\in \{a,b,c,d\}$ and then it left $F$ for the last time to another vertex $y\in \{a,b,c,d\}$. But $H'[F'\cup \{a,b,c,d\}]$ is strongly connected, so there is a path from $x$ to $y$ and hence also from $s$ to $t$ in $H'$. 
	
To finish the proof of the claim we show that if $H'$ is not an inclusion-minimal solution to $R$, then also $H$ is not an inclusion-minimal solution to $R$. Let $uv$ be an arc in $H'$ such that $H'-uv$ is also a solution to $R$. We distinguish 2 cases:
	
\begin{enumerate}
	\item[1)] The arc $uv$ is in $H-F$ as well.
\end{enumerate}	
	 
We claim that in this case $H-uv$ is also a solution to $R$. To prove this claim, let $s$ and $t$ be two terminals and $P'$ a path from $s$ to $t$ in $H'-uv$. We show that there is a path from $s$ to $t$ in $H-uv$ as well. If $P'$ does not contain a vertex from $F'$, then whole $P'$ is in $H-uv$ as well and we are done. Otherwise, let us split $P'$ in three parts as follows. Let $x$ be the first vertex on $P'$ in $\{a,b,c,d\}$ and let $y$ be the last vertex on $P'$ in $\{a,b,c,d\}$. Since $H[F\cup \{a,b,c,d\}]$ is strongly connected, there is a path $Q$ from $x$ to $y$ in $H[F\cup \{a,b,c,d\}]$. Then the path $s[P']x\circ Q\circ y[P']t$ is in $H-uv$ and the case 1) follows. 
	 
\begin{itemize}
	\item[2)] The arc $uv$ is not in $H-F$.
\end{itemize}

Since $H'[F'\cup \{a,b,c,d\}]$ is a ladder, from Lemma~\ref{lem:ladder_properties} 
it follows that $H'[F'\cup \{a,b,c,d\}]$ is an inclusion-minimal solution connecting $a,b,c,d$ and in particular it is an union of two directed paths $P'_1$ from $a$ to $d$ and $P'_2$ from $c$ to $b$. Therefore, $H'[F'\cup \{a,b,c,d\}]-uv$ either does not connect $a$ to $d$ or $c$ to $b$.
First, if $H'[F'\cup \{a,b,c,d\}]-uv$ does not contain a directed path from $a$ to $d$, then we claim that for any edge $e\in A(P_1)\setminus A(P_2)$, $H-e$ is also a solution to $R$. 
Again, let $s$ and $t$ be two terminals and $P'$ a path from $s$ to $t$ in $H'-uv$. If $P'$ does not contain any vertex of $F'$, then $P'$ is in $H-e$ as well and the case follows. Otherwise $P'$ enters $H'[F\cup\{a,b,c,d\}]$ in either $a$ or $c$ (note that only in-neighbor of $b$ is $a$ and of $d$ is $c$) and leaves $H'[F\cup\{a,b,c,d\}]$ in $b$ or $d$. However, if it enters in $a$ and leaves in $b$ or enters in $c$ and leaves in $d$, then the whole portion of the path from $a$ to $b$ or from $c$ to $d$, respectively, can be replaced by an arc. The resulting path is then whole in $H'-F'$ and hence also in $H-e$. Otherwise, since $H'[F\cup\{a,b,c,d\}]$ does not contain a path from $a$ to $d$, $P'$ enters $H'[F\cup\{a,b,c,d\}]$ in $c$ and leaves in $b$. However, in this case the path $s[P']c\circ P_2\circ b[P']t$ is in $H-e$ and hence indeed $H-e$ is a solution to $R$, contradiction. Finally, the case  if $H'[F'\cup \{a,b,c,d\}]-uv$ does not contain a directed path from $c$ to $b$ is symmetric. 
\end{claimproof}

Now, if $\sym(H)$ does not contain a $k\times k$ grid for any $k\ge 10$, then we done. Otherwise, let $\ell$ be such that $\sym(H)$ contains $\ell\times\ell$ grid, but does not contain $(\ell+1)\times(\ell+1)$ grid as a minor. Let us denote by $K$ an $\ell\times\ell$ grid minor of $\sym(H)$ and let $f\colon V(K)\rightarrow 2^{H}$ be the minor model of $K$ in $\sym(H)$, that is, for all $u,v\in V(K)$ such that $u\neq v$ it holds that 
\begin{enumerate}
\item $f(u)\cap f(v)=\emptyset$,
\item if  $\{u,v\}\in E(K)$, then there are vertices $x,y$ such that $x\in f(u), y\in f(v), \{x,y\}\in E(\sym(H))$, and 
\item $\sym(H)[f(u)]$ is a connected graph. 
\end{enumerate}
It is well known that a minor model of $K$ in $\sym(H)$ exists if and only if $K$ is a minor of $\sym(H)$ (see~\cite[Proposition 1.7.2]{Diestel17}).

Observe that at most $4$ vertices of $K$ map by $f$ to one of vertices in $\{a,b,c,d\}$ and these vertices form a separator between the vertices of $K$ that map on subsets of $F$ and the vertices that maps on subsets of $V(H) \setminus F$. Note, that treewidth of $\sym(H)[F]$ is at most $2$ and hence $\sym(H)[F]$ does not contain $3\times 3$ grid as a minor. Note that for each separator of size at most $4$ in a $10 \times 10$ grid, one of the connected components obtained after removing the separator contains all but at most $4$ vertices and a large portion of the grid preserved. Hence, as $K$ is at least a $10 \times 10$ grid and we have a separator of size at most $4$, at most $4$ of vertices of $K$ can map by $f$ on a subset of $F$. 

Let us denote by $K_F$ the subgraph of $K$ such that $v\in V(K_F)$ if and only if $f(v)\cap (F\cup \{a,b,c,d\})\neq \emptyset$ and $\{u,v\} \in E(K)$ is in $E(K_F)$ if and only if there are vertices $x,y \in (F\cup \{a,b,c,d\})$ such that $x\in f(u), y\in f(v), \{x,y\}\in E(\sym(H))$. From the discussion above $|V(K_F)|\le 8$. Furthermore, each vertex of $K_F$ with a neighbor outside of $K_F$ in $K$ has to have in its map at least one of vertices $\{a,b,c,d\}$. It follows that $K_F$ is a minor of $\sym{H}[F\cup \{a,b,c,d\}]$, where $f$ with its image restricted to $F\cup \{a,b,c,d\}$ is the minor model.
Hence, if we obtain $H'$ by replacing $H[F]$ by any $H'[F']$ such that $\sym(H')[F'\cup \{a,b,c,d\}]$ has $K_F$ as a minor such that there is a minor model $g\colon K_F\rightarrow 2^{F'\cup \{a,b,c,d\}}$ with the same pre-images of subsets  containing $a$,$b$,$c$, and $d$ as $f$, then $\sym(H')$ contains $\ell\times\ell$ minor as well. 

Note that $\sym(H)[F\cup \{a,b,c,d\}]$ is a minor of the underlying undirected graph of some ladder $G_{n,\emptyset}$ and hence also $K_F$ is a minor of $\sym(G_{n,\emptyset})$. We claim there is a universal constant $n'$ such that $K_F$ is a minor of $\sym(G_{n',\emptyset})$. This follows from the following argument.
It is easy to express by an MSO formula the fact that a $t$-boundaried graph contains the boundaried graph $K_F$ as a minor (see, e.g., Preliminaries of Fomin et. al~\cite{FominLMS12}). Furthermore, it follows from Lemma~\ref{lem:outerplanar_ladder} that we can say that a boundaried graph $B$, with boundary vertices $a,b,c,d$ is isomorphic to some $\sym(G_{n',\emptyset})$ by just saying that it is outerplanar, $2$-connected, vertices $a,b,c,d$ have all degree precisely two and all other vertices have degree precisely 
three. Since all of these properties are MSO definable, it follows that we can express the fact that a graph $B$ is isomorphic with an underlying undirected graph of some ladder $G_{n',\emptyset}$ with a constant size MSO formula as well. Hence, there exists a constant size MSO formula $\phi$ that express the fact that a boundaried graph $G$ is a an underlying undirected graph of a ladder and has $K_F$ as a minor. Since we know that $\sym(G_{n,\emptyset})$ is a model of $\phi$ and $\phi$ has a constant size, it follows that a size of a minimum sized model is at most some constant that depends only on the length of the formula. The only thing we need to be careful about is that we do not want the genus of graph to grow after replacing $H[F]$ with the smallest model of $\phi$. This can be simply done by preserving the parity of the length of the ladder. It is easy to adjust the minor model to reflect that.  Hence, we might take a one longer ladder, if the orientation of $c$ and $d$ in $H$ does not match the one in the ladder, we found.  
This finishes the proof of the lemma.
\end{proof}
}

\begin{lemma}\appmark\label{lem:reducing_length}
	There exists a directed graph $H'$ such that 
	\begin{itemize}
		\item the genus of $\sym(H')$ is at most the genus of $\sym(H)$,
		\item $T\subseteq V(H')$, 
		\item For all $s,t\in T$, there is a directed path from $s$ to $t$ in $H - (T\setminus\{s,t\})$ if and only if there is a directed path from $s$ to $t$ in $H' - (T\setminus\{s,t\})$,
		\item $H'$ is an inclusion-minimal solution to $R$,
		\item $\tw(\sym(H)) \le 20^{4(4g+3)}\tw(\sym(H'))$, and
		\item for any arc $st\in A(R)$ there is a path $P$ in $H'$ with length at most $\OhOp{|I_P|}$.
	\end{itemize}
\end{lemma}
\sv{
\begin{proof}[Proof sketch]
We obtain $H'$ by recursively applying Lemma~\ref{lem:protrusion_replacement} until there is no ladder with the boundary of size at most $4$ that can be shortened by applying Lemma~\ref{lem:protrusion_replacement}.
It follows from Lemma~\ref{lem:ladder_structure_between_important_points} that distance between any two consecutive $p_i,p_j \in Q_P \cup I_P$ is constant.
Since the genus of $\sym(H)$ is at most $g$, it follows from Proposition~\ref{pro:genus-minor} that $\sym(H)$ is $K_{3,4g+3}$-minor-free.
Therefore, due to Proposition~\ref{thm:MinorFreeTW}, the treewidth of $\sym(H)$ is at most $20^{4(4g+3)}\ell$, where $\ell$ is the size of the largest grid minor of $\sym(H)$.
\end{proof}
}
\toappendix{
\begin{proof}[Proof of Lemma~\ref{lem:reducing_length}]
	We obtain $H'$ by recursively applying Lemma~\ref{lem:protrusion_replacement} until there is no ladder with the boundary of size at most $4$ that can be shortened by applying Lemma~\ref{lem:protrusion_replacement}. 
	 The graph $H'$ always exists and is well defined as each successful application of Lemma~\ref{lem:protrusion_replacement} decrease the size of the graph. From the same lemma it follows that for every $k\ge 10$, if $\sym(H)$ contains $k\times k$ grid as a minor, then $\sym(H')$ contains $k\times k$ grid as a minor. Now let $\ell$ be the maximum positive integer such that $\sym(H)$ contains an $\ell\times\ell$ grid as a minor. 
	 Since the genus of $\sym(H)$ is at most $g$, it follows from Proposition~\ref{pro:genus-minor} that $\sym(H)$ is $K_{3,4g+3}$-minor-free. Therefore, due to Proposition~\ref{thm:MinorFreeTW}, the treewidth of $\sym(H)$ is at most $20^{4(4g+3)}\ell$. If $\ell\le 10$, then clearly $\tw(\sym(H)) = O(\tw(\sym(H')))$. Otherwise, $\sym(H')$ also contains an $\ell\times\ell$ grid as a minor. However, this means that the treewidth of $\sym(H')$ is at least $\ell$ and hence $\tw(\sym(H)) = O(\tw(\sym(H')))$. The inclusion-minimality and the fact that $H'$ connects precisely the same set as $H$ follows from Lemma~\ref{lem:protrusion_replacement} and the fact that we were replacing only subgraphs that did not contain any terminal and we preserved all the connection between vertices on the boundary. 
	 
	Now let $s$ and $t$ be two terminals such that $st\in A(R)$ and let $P$ be a directed $s$-$t$ path in $H'$. Clearly, $|I_P\cup Q_P|\le 5|I_P|$.	We will show that between any two consecutive vertices in $I_P\cup Q_P$ on the path $P$ there are at most constantly many other vertices.  	
	Let $p_i, p_j\in V(P)$, $1\le i<j\le r$ 
	be two vertices such that $\{p_i,p_j\}\in Q_P\cup I_P$ and for all $k$ such that $i<k<j$ it holds that $p_k\notin Q_P\cup I_P$. If $j-i\le 5$, then we are fine and there are only constantly many vertices between $p_i$ and $p_j$. Otherwise, let $F= \{p_{i+1},\ldots, p_{j-1}\}$. 
	It follows from Lemma~\ref{lem:ladder_structure_between_important_points} that 
	if $C$ is the connected component of $H-\{p_{i+1},p_{i+2}, p_{j-2}, p_{j-1}\}$ containing $p_{i+3}$, then $H[C\cup \{p_{i+1},p_{i+2}, p_{j-2}, p_{j-1}\}]$ is a ladder. However, $F\subseteq C$ and since we cannot apply  Lemma~\ref{lem:protrusion_replacement} it follows that it has a constant size. Therefore, there are at most constantly many vertices between $p_i$ and $p_j$ for any two consecutive vertices in $I_P\cup Q_P$. It follows that length of $P$ is at most $\OhOp{|I_P\cup Q_P|}=\OhOp{|I_P|}$ and the lemma follows. 
\end{proof}
}

\begin{proof}[Proof of Theorem~\ref{thm:H_has_bounded_tw}]
By Lemma~\ref{lem:reducing_length}, there exists $H'$ connecting the same set of terminals as $H$ such that $\tw(\sym(H)) \le 20^{4(4g+3)}{\tw(\sym(H'))}$ and for each arc $st\in A(R)$, there is a directed path from $s$ to $t$ of length at most $\OhOp{|I_P|}$ in $H'$. Furthemore, all the vertices of $H'$ are on some path of length at most $\OhOp{|I_P|}$ between two terminals in $H'$. 
By Lemma~\ref{lem:small_distances}, it follows that there is a path of length at most $\OhOp{q}$ between each pair of terminals in $\sym(H')$ and hence the diameter of $\sym(H')$ is also at most $\OhOp{q}$. Finally by Proposition~\ref{prop:genus_tw}, it follows that $\sym(H')$ has treewidth $\OhOp{gq}$, where $g$ is the genus of $\sym(H')$. Since the genus of $\sym(H')$ is at most genus of $\sym(H)$, which is constant, the lemma follows. 
\end{proof}

\section{Improved ETH-based Lower Bound for General Graphs}\label{sec:DSN_lowerbounds}
\sv{\toappendix{
\section{Additions to Section~\ref{sec:DSN_lowerbounds}}
}}
Our proof is based on a reduction from (a special case of) the following problem:
\prob{\textsc{Partitioned Subgraph Isomorphism} (PSI)}
{Two unoriented graphs $G$ and $H$ with $|V(H)| \le |V(G)|$ ($H$ is \emph{smaller}) and a mapping $\psi\colon V(G) \to V(H)$.}
{Is $H$ isomorphic to a subgraph of $G$? I.e., is there an injective mapping $\phi\colon V(H) \to V(G)$ such that $\{\phi(u),\phi(v)\} \in E(G)$ for each $\{u,v\} \in E(H)$ and $\psi \circ\phi$ is the identity?}

\lv{We use the following theorem.}

\begin{theorem}[{Marx~\cite[Corollary 6.1]{Marx10}}]
If there is a recursively enumerable class $\mathcal{H}$ of graphs with
unbounded treewidth, an algorithm $\mathbb{A}$, and an arbitrary
function $f$ such that $\mathbb{A}$ correctly decides every instance
of \textsc{Partitioned Subgraph Isomorphism} with the smaller graph $H$ in
$\mathcal{H}$ in time $f(H)n^{o(\textup{tw}(H)/\log \textup{tw}(H))}$,
then ETH fails.
\end{theorem}
  
It is known that there are infinitely many 3-regular graphs such that each such graph $H$ has treewidth $\Theta(|V(H)|)$ (cf.~\cite[Proposition~1, Theorem~5]{GroheM09}).
Using the class of 3-regular graphs as $\mathcal{H}$ in the above theorem, we arrive at the following corollary.
\begin{corollary}\label{cor:psi_hard}
If there is an algorithm $\mathbb{A}$ and an arbitrary
function $f$ such that $\mathbb{A}$ correctly decides every instance
of \textsc{Partitioned Subgraph Isomorphism} with the smaller graph $H$ being 3-regular in time $f(H)n^{o(V(H)/\log V(H))}$, then ETH fails.
\end{corollary}

Our plan is to use this corollary. To this end we transform the (special) instances of PSI to instances of DSN.

\begin{constr}\label{con:hardness}
Let $(G,H, \psi)$ be an instance of PSI with $H$ 3-regular and denote \mbox{$k=|V(H)|$}. Note that then $|E(H)|=O(k)$.
We let $r=\left\lceil\sqrt{k}\right\rceil$.
We first compute labellings $\alpha\colon V(H) \to X$, $\beta\colon V(H) \to Y$, and $\gamma\colon E(H) \to Z$, 
where ${X=\{x_1, \ldots, x_{\max}\}}, {Y=\{y_1, \ldots, y_{\max}\}}$, and $Z=\{z_1, \ldots, z_{\max}\}$ are three new sets.
We want to keep the sets $X,Y,Z$ of size $O(r)$ while fulfilling the following constraints:
\begin{enumerate}[(i)]
 \item $\forall u,v \in V(H): (\alpha(u) \neq \alpha(v)) \vee (\beta(u) \neq \beta(v))$ \label{cond:each_vert_diff}
 \item $\forall \{u,v\} \in E(H): (\alpha(u) \neq \alpha(v)) \wedge (\beta(u) \neq \beta(v))$ \label{cond:each_adj_diff}
 \item $\forall e,f \in E(H), \forall u,v\in V(H): ((u \in e) \wedge (v \in f) \wedge (\alpha(u)=\alpha(v))) \implies (\gamma(e) \neq  \gamma(f))$.\label{cond:each_edge_diff}
\end{enumerate}
In other words, the pair $(\alpha(u),\beta(u))$ uniquely identifies the vertex $u$, adjacent vertices share no labels and both pairs $(\alpha(u),\gamma(\{u,v\}))$ and $(\alpha(v),\gamma(\{u,v\}))$ uniquely identify the edge $\{u,v\}$.

To obtain such labelling, first colour the vertices of $H$ greedily with colours $1, \ldots, 4$, denote $\eta$ the colouring and $A_1, \ldots, A_4$ the set of vertices of colour $1, \ldots, 4$, respectively. For every $i \in \{1, \ldots,4\}$, we split the set $A_i$ into sets $A_{i,1},\ldots, A_{i,a_i}$ such that for every $j \in \{1,\ldots, a_i-1\}$ the set $A_{i,j}$ is of size $r$ and the set $A_{i,a_i}$ is of size at most~$r$. Since $r=\left\lceil\sqrt{k}\right\rceil$ we know that there will at most $r$ sets of size $r$ and, thus, at most $r+4$ sets in total.
We assign to each nonempty set $A_{i,j}$ a unique label $x_\ell$ and let $\alpha(u)=x_\ell$ for every $u \in A_{i,j}$. Note that $|X| \le r+4$.

Next construct a graph $H'$ from $H$ by turning each $A_{i,j}$ into a clique. Since the degree of each vertex in $H$ is 3 and the size of each $A_{i,j}$ is at most $r$, the degree of each vertex in $H'$ is at most $r+2$. Hence we can colour the vertices of $H'$ greedily with colours $y_1, \ldots, y_{r+3}$ and let $\beta$ be the colouring.

Finally, we construct a multigraph $H''$ from $H'$ by contracting each clique $A_{i,j}$ to a single vertex. We keep multiple edges between two vertices if they are a result of the contraction, but we remove all loops. Note that the edges preserved are exactly the edges of $H$. Since the size of each $A_{i,j}$ is at most $r$ and $H$ is 3-regular, the maximum degree (counting the multiplicities of the edges) is at most $3r$. Therefore, the maximum degree in the line graph $L(H'')$ of $H''$ is at most $6r-2$. Thus, we can colour the edges of $H''$ greedily with colours $z_1, \ldots, z_{6r-1}$ and let $\gamma$ be the colouring.

Let us check that the labellings fulfill the constraints. First, if $\alpha(u)=\alpha(v)$, then $\{u,v\} \in E(H')$ and, thus, $\beta(u) \neq \beta(v)$. If $\{u,v\} \in E(H)$, then $\{u,v\} \subseteq A_{i,j}$ would imply that $u$ and $v$ are colored by the same colour by $\eta$ --- a contradiction. Hence, $\alpha(u) \neq \alpha(v)$ and, since $E(H) \subseteq E(H')$, we also have $\beta(u) \neq \beta(v)$. Finally, if $e=\{u,v'\}$, $f =\{u',v\}$, and $\alpha(u) = \alpha(v)$, then the edges $e$ and $f$ share a vertex in $H''$ and, thus, $\gamma(e) \neq  \gamma(f)$.

Note also that the labellings can be obtained in $O(|V(H)|^2)$ time.

Having the labellings at hand, we construct the instance $(G',R)$ of DSN as follows (refer to Figure~\ref{fig:reductionSchema} for an overview of the construction). We let $V(G')=V \cup W \cup X \cup Y \cup Z$, where $V=V(G)$, $W=\{w_{uv} \mid \{u,v\} \in E(G)\}$, and $X,Y,Z$ are the images of $\alpha, \beta, \gamma$ as defined previously. We let $T=V(R)=X \cup Y \cup Z$. Note that $q=|T|=O(r)=O(\sqrt{k})$.
We let $A(G')=A_V \cup A_W$, where $A_V=\left\{ \left( \alpha(\psi(u)),u),(u,\beta(\psi(u)) \right) \mid u \in V \right\}$ and $A_W=\{(u,w_{uv}),(v,w_{uv}),(w_{uv}, \gamma(\{\psi(u),\psi(v)\}))\mid \{u,v\} \in E(G)\}$. We assign unit weights to all arcs of $G'$.
Finally let $A(R)=A_Y\cup A_Z$, where $A_Y=\{(\alpha(u),\beta(u)) \mid u \in V(H)\}$ and $A_Z=\{(\alpha(u), \gamma(\{u,v\})), (\alpha(v),\gamma(\{u,v\})) \mid \{u,v\} \in E(H)\}$.

Let us stop here to discuss the size of $A(R)$. By Condition~(\ref{cond:each_vert_diff}) on the labellings we have $|A_Y|=|V(H)|$. By Condition~(\ref{cond:each_adj_diff}) we have $(\alpha(u), \gamma(\{u,v\})) \neq (\alpha(v),\gamma(\{u,v\}))$ for any $\{u,v\} \in E(H)$. Hence, by Condition~(\ref{cond:each_edge_diff}) the size of $A_Z$ is exactly $2|E(H)|$.
\end{constr}

\begin{figure}[bt]
\usetikzlibrary{calc}

\begin{tikzpicture}[node distance=1.4cm]
\tikzstyle{vertex}=[circle, fill=black,draw, inner sep=2pt]
\tikzstyle{terminal}=[fill=black, inner sep=4pt]
\tikzstyle{nonTerminal}=[circle, draw, inner sep=3pt]
\tikzstyle{edgeR}=[dashed, ->, thick]
\tikzstyle{edgeG}=[->,thick]
\tikzstyle{edgeH}=[thick]

\begin{scope}[local bounding box=H]
\node[label={180:$\alpha(u'), \beta(u')$}, label={90:$u'$}, vertex] (u) at (0,-.7) {};
\node[label={180:$\alpha(v'), \beta(v')$},label={270:$v'$},vertex, below of=u] (v) {};

\draw[edgeH] (u) to node[midway,right] {$\gamma(\{u',v'\})$}(v);
\end{scope}
\node at ($(H) - (0,3.2)$) (HH) {$H$};

\begin{scope}[local bounding box=G1, xshift=3cm]
\node[label={90:$u$}, vertex] (u) at (0,-.7) {};
\node[label={270:$v$},vertex, below of=u] (v) {};

\draw[edgeH] (u) to (v);
\end{scope}
\node at ($(HH -|G1)$) (G11) {$G$};

\begin{scope}[local bounding box=G, xshift=6cm]
\node[label={180:$\alpha(\psi(u))$}, terminal] (uUP) {};
\node[label={180:$u$}, nonTerminal, below of=uUP] (uMID) {};
\node[label={180:$\beta(\psi(u))$}, terminal, below of=uMID] (uDOWN) {};

\draw[edgeG] (uUP) to (uMID);
\draw[edgeG] (uMID) to (uDOWN);
\draw[edgeR] (uUP) to[out=220, in=130] (uDOWN);

\node[label={0:$\alpha(\psi(v))$}, terminal] (vUP)  at ($(uUP) + (3,0)$) {};
\node[label={0:$v$}, nonTerminal, below of=vUP] (vMID) {};
\node[label={0:$\beta(\psi(v))$}, terminal, below of=vMID] (vDOWN) {};

\draw[edgeG] (vUP) to (vMID);
\draw[edgeG] (vMID) to (vDOWN);
\draw[edgeR] (vUP) to[out=-40, in=50] (vDOWN);

\node[label={$w_{uv}$}, nonTerminal] at ($(uMID)!.5!(vMID) - (0,1)$) (uv) {};
\node[label={0:$\gamma(\{\psi(u),\psi(v)\})$}, terminal, below of=uv] (uvDOWN) {};
\draw[edgeG] (uv) to (uvDOWN);

\draw[edgeG] (uMID) to (uv);
\draw[edgeG] (vMID) to (uv);

\draw[edgeR] (uUP) to[out=-40] (uvDOWN);
\draw[edgeR] (vUP) to[out=220, in=40] (uvDOWN);
\end{scope}
 \node at ($(HH -|G)$) {$G'$ and $R$};
\end{tikzpicture}
\caption{\label{fig:reductionSchema}
An ilustration of Construction \ref{con:hardness}. Left is a pattern graph $H$, middle a host graph $G$ and right the produced graphs $G'$ and $R$ combined. We assume $\psi(u)=u'$ and $\psi(v)=v'$ here. On the right the terminals are depicted by full sqares and non-terminals by empty circles. Arcs in $G'$ are drawn solid, while the arcs of $R$ are dashed.
}
\end{figure}
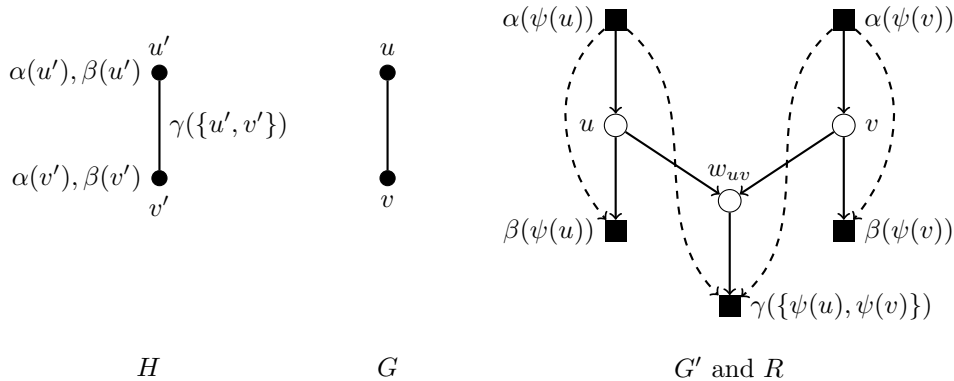

Next we show that the construction transform yes-instances of PSI to instances of DSN with bounded value of the optimum.

\begin{lemma}\label{lem:cor:psi_to_dsn}
If there is $\phi$ forming a solution to the instance $(G,H, \psi)$ of PSI, then there is subgraph $P$ of $G'$ forming a solution to the instance $(G',R)$ of DSN with cost $|A(P)| \le 2|V(H)|+3|E(H)|$.
\end{lemma}

\begin{proof}
Let $\phi$ be a solution to the instance $(G,H, \psi)$.
Since $\phi$ is a solution, we know that  $\{\phi(u),\phi(v)\} \in E(G)$ whenever $\{u,v\} \in E(H)$.
Consider the subgraph $P=G'[V_\phi]$ of $G'$ induced by $V_\phi=X \cup Y \cup Z \cup V' \cup W'$, where $V' = \{\phi(v) \mid v \in V(H)\}$ and $W'=\{w_{\phi(u)\phi(v)}\mid \{u,v\} \in E(H)\}$. Obviously $|V'|=|V(H)|$ and $|W'|=|E(H)|$. 

Since each arc in $A_W$ is incident to some vertex in $W$ and each vertex in $W$ is incident to exactly 3 such arcs, $P$ contains at most $3|E(H)|$ arcs from $A_W$. Similarly, since each arc in $A_V$ is incident to some vertex in $V$ and each vertex in $V$ is incident to exactly 2 such arcs, $P$ contains at most $2|V(H)|$ arcs from $A_V$. Thus, $P$ contains at most $2|V(H)|+3|E(H)|$ arcs in total.

We want to show for each $(s,t) \in A(R)$ that there is a directed path from $s$ to $t$ in $P$. Indeed, if $(x,y) \in A_Y$, then $x=\alpha(u)$ and $y=\beta(u)$ for some $u \in V(H)$ and $\alpha(u),\phi(u),\beta(u)=\alpha(\psi(\phi(u))),\phi(u),\beta(\psi(\phi(u)))$ is a path of length 2 from $x$ to $y$ in $P$. If $(x,z) \in A_Z$, then $x=\alpha(u)$ and $z=\gamma(\{u,v\})$ for some $\{u,v\} \in E(H)$ and $\alpha(u),\phi(u),w_{\phi(u)\phi(v)},\gamma(\{u,v\})$ is a path of length 3 from $x$ to $z$ in $P$.
This finishes the proof.
\end{proof}

Next we show that the value of the optimum of the instances of DSN produced by the construction can only be appropriately bounded if we started with a yes-instance of PSI.

\begin{lemma}\appmark\label{lem:cor:dsn_to_psi}
If there is subgraph $P$ of $G'$ forming a solution to the instance $(G',R)$ of DSN with cost $|A(P)| \le 2|V(H)|+3|E(H)|$, then there is $\phi$ forming a solution to the instance $(G,H, \psi)$ of PSI.
\end{lemma}
\toappendix{
\begin{proof}[Proof of Lemma~\ref{lem:cor:dsn_to_psi}]
Let $P$ be a solution to the instance $(G',R)$ of DSN with cost $|A(P)| \le 2|V(H)|+3|E(H)|$. Without loss of generality let us assume that $P$ contains no isolated vertices.

Note that any path from a vertex in $X$ to a vertex in $Y$ in $G'$ is of length exactly 2 with the middle vertex in $V$. Moreover, each vertex $v$ in $V$ only lies on a path from $\alpha(\psi(v))$ to $\beta(\psi(v))$. Hence, for each $(x,y) \in A(R)$ with $x \in X$ and $y \in Y$ the vertex $v_{xy}$ on the $x$-$y$-path in $P$ is not on a $x'$-$y'$-path in $P$ for any $(x',y') \in A(R)$ with $x' \in X$, $y' \in Y$, and $(x,y)\neq (x',y')$. Therefore $P$ contains at least $2|A_Y|$ arcs from $A_V$, which is at least $2|V(H)|$. 

Similarly, any path from a vertex in $X$ to a vertex in $Z$ in $G'$ is of length exactly 3 with the second vertex in $V$ and the third in $W$. Moreover, each vertex $w_{uv}$ in $W$ only lies on  paths from $\alpha(\psi(u))$ to $\gamma(\{\psi(u),\psi(v)\})$ and from $\alpha(\psi(v))$ to $\gamma(\{\psi(u),\psi(v)\})$. Hence, each vertex in $W$ is on at most two paths in $P$ required by $R$ and we have to include in $P$ 2 arcs (3 arcs) from $A_W$ incident to it, if it is on 1 such path (2 such paths), respectively. Hence, $P$ contains at least $\frac32|A_Z|$ arcs from $A_W$, which is at least $3|E(H)|$.

Thus, the budget is tight, and $P$ includes exactly the minimum number of arcs from each group. In particular, if $P$ contains some vertex $w_{uv}$ from $W$, then it contains all 3 arcs incident to this vertex. Similarly, if $P$ contains a vertex $v \in V$, then $P$ must contain also some arc from $A_V$ incident to it and, thus, both arcs incident to it, as otherwise there would be too many arcs from $A_V$ in $P$. Then it is the only vertex in $\psi^{-1}(\psi(v))$ in $P$.

For each $v \in V(H)$ let $\phi(v)$ be the unique vertex in $\psi^{-1}(v)$ contained in $P$, that is the unique vertex on a path from $\alpha(v)$ to $\beta(v)$ in $P$. We claim that $\phi$ is a solution to the instance $(G,H, \psi)$ of PSI. By the way we constructed $\phi$ we get that $\psi \circ\phi$ is the identity and, hence, $\phi$ is necessarily injective.

Let now $\{u,v\} \in E(H)$. Since $A(R)$ contains the arcs $(\alpha(u), \gamma(\{u,v\})), (\alpha(v),\gamma(\{u,v\}))$, there must be a vertex $w_{u'v'}$ in $V(P)\cap W$ such that $P$ contains the path $\alpha(u),u',w_{u'v'}, \gamma(\{u,v\})$ and the path $\alpha(v),v',w_{u'v'}, \gamma(\{u,v\})$. It follows that $u' =\phi(u)$, $v'=\phi(v)$ and $\{\phi(u),\phi(v)\}=\{u',v'\}$ is an edge of $G$ as required.
This finishes the proof.
\end{proof}
}

\lv{
Now we are ready to give the proof of Theorem~\ref{thm:gen_hard}.
}

\begin{proof}[Proof of Theorem~\ref{thm:gen_hard}]
Let $\mathbb{A}$ be an algorithm that correctly solves DSN (on general graphs) in time $f(q)n^{o(q^2/\log q)}$ for some function $f$.
Let us construct an algorithm $\mathbb{B}$ for PSI with the smaller graph $H$ being 3-regular as follows:
Let $(G,H, \psi)$ be an instance of PSI with $H$ 3-regular. We use Construction~\ref{con:hardness} to construct the instance $(G',R)$ of DSN.
Then run $\mathbb{A}$ on $(G',R)$ and return yes if and only if the cost of the obtained solution $P$ is $|A(P)| \le 2|V(H)|+3|E(H)|$.
\lv{Otherwise return no.}
The answer of $\mathbb{B}$ is correct by Lemmata~\ref{lem:cor:psi_to_dsn} and~\ref{lem:cor:dsn_to_psi}.

Let us analyze the running time of $\mathbb{B}$. Let us denote $k=|V(H)|$ and $n=|V(G)|$. We may assume that $k \le n$ as otherwise we can immediately answer no. The labellings can be obtained in $O(k^2)$ time. Graph $G$ has at most $O(n^2)$ edges and the graphs $G'$ and $R$ can be constructed in linear time in the number of vertices and edges of the graphs $G$ and $H$, respectively. That is, Construction~\ref{con:hardness} can be performed in $O(n^2)$ time and in particular $G'$ has $O(n^2)$ vertices. However, by the construction, the number $q$ of vertices of graph $R$ is $O(\sqrt{k})$. Now, $\mathbb{A}$ runs on $(G',R)$ in time $f(q)|V(G')|^{o(q^2/\log q)}=f'(\sqrt{k})n^{o((\sqrt{k})^2/\log \sqrt{k})}=f''(k)n^{o(k/\log k)}$ for some functions $f,f'$, and $f''$. But then the whole $\mathbb{B}$ runs in $f''(k)n^{o(k/\log k)}$ time and ETH fails by Corollary~\ref{cor:psi_hard}, finishing the proof.
\end{proof}

\section{Conclusions}
Our results show that we can solve DSN in time $n^{\OhOp{q}}$ when the input directed graph is embeddable on a fixed genus surface\lv{; the constant hidden in the $O$-notation indeed depends on the genus}.
However, for general graphs it is unlikely to obtain even an algorithm running in time $n^{o(q^2/\log(q))}$.
It would be interesting to see what happens for the graph classes that are somewhere in between.
For example, it is not difficult to show that the graph $H'$ that we obtain in Section~\ref{sec:DSN_on_surface} has at most $\OhOp{q^3}$ vertices and hence the size of the largest grid minor of $H'$ is of size at most $\OhOp{q^{3/2}}\times \OhOp{q^{3/2}}$.
Therefore, with a careful modification of our approach, one can show that there is an $n^{\OhOp{q^{3/2}}}$ time algorithm for DSN when the input graph excludes a fixed minor. 
However, it remains open whether the running time $n^{\OhOp{q^{3/2}}}$ is asymptotically optimal or whether it is possible to design an $n^{\OhOp{q}}$ time algorithm for DSN in this case.



\bibliography{ref}

%
%

\end{document}